\documentclass{article}[10pt]
\usepackage{epsfig,pict2e}
\usepackage{amsmath,tabu}
\usepackage{amsthm}
\usepackage{amsfonts}
\usepackage{amssymb}
\usepackage{mathtools}
\usepackage{relsize}
\usepackage{tikz}
\usepackage{tikz-cd}
\usetikzlibrary{calc} 
\usepackage{stmaryrd}
\usepackage{geometry}
\usepackage[toc,page]{appendix}
\usepackage{graphicx}
\usepackage{caption}
\usepackage{subcaption}
\usepackage{aliascnt}  
\usepackage{hyperref}

\usepackage{array} 
\usepackage{paralist} 
\usepackage{verbatim} 
\usepackage{float}
\newtheorem{theorem}{Theorem}[section]
\newtheorem{proposition}[theorem]{Proposition}

\newaliascnt{lemma}{theorem}  
\newtheorem{lemma}[lemma]{Lemma}  
\aliascntresetthe{lemma}  
\newaliascnt{definition}{theorem}  
\newtheorem{definition}[definition]{Definition}
\aliascntresetthe{definition}  

\newaliascnt{example}{theorem}  
\newtheorem{example}[example]{Example}
\aliascntresetthe{example}  

\usepackage{accents}
\newcommand{\ubar}[1]{\underaccent{\bar}{#1}}

\newenvironment{indented}{\begin{indented}}{\end{indented}}
\def\indented{\list{}{\itemsep=0\p@\labelsep=0\p@\itemindent=0\p@
   \labelwidth=0\p@\leftmargin=\mathindent\topsep=0\p@\partopsep=0\p@
   \parsep=0\p@\listparindent=15\p@}\footnotesize\rm}

\numberwithin{equation}{section}

\begin{document}
\begin{center}
\LARGE{\textbf{Singularity confinement in delay-differential Painlev\'e equations}}
\end{center}
\vspace{1ex}
\begin{center}
\normalsize{Alexander Stokes}\\
\noindent \small{Department of Mathematics, University College London,\\
Gower Street, London, WC1E 6BT, UK}
\end{center}
\normalsize


\begin{abstract} 
We study singularity confinement phenomena in examples of delay-differential Painlev\'e equations, which involve shifts and derivatives with respect to a single independent variable. We propose a geometric interpretation of our results in terms of mappings between jet spaces, defining certain singularities analogous to those of interest in the singularity analysis of discrete systems, and what it means for them to be confined. For three previously studied examples of delay-differential Painlev\'e equations, we describe all such singularities and show they are confined in the sense of our geometric description.
\end{abstract}

\section{Introduction}
Singularity confinement is a phenomenon first proposed as an integrability criterion for discrete systems \cite{RGP1991}, and has been used to great effect to obtain discrete analogues of the Painlev\'e differential equations \cite{RGWSTRIHOMOGRAPHIC,RGH1991, RGdPs}. Its geometric interpretation has led to novel connections between discrete integrable systems and birational algebraic geometry, most notably Sakai's geometric framework and classification scheme for discrete Painlev\'e equations \cite{SAKAI2001}. We study delay-differential equations, for which a kind of singularity confinement test has been used to isolate integrability candidates and obtain delay-differential equations of Painlev\'e-type \cite{RGMOREIRA,RGTAMIZ}. These so-called delay Painlev\'e equations possess analogues of many integrability properties of their discrete and differential counterparts, and it is natural to ask whether a geometric theory may be developed for them. \\

Compared to the discrete case, the understanding of singularity confinement in this class of equations is in its infancy. In particular, we do not have available to us the definition of singularity confinement in second-order discrete systems as the iteration mappings of the systems lifting to isomorphisms between rational surfaces. Further, even for heuristic observations in the absence of a proper definition of confinement, the presence of derivatives leads to challenges, as different multiplicities with which solutions take singular values lead to infinitely many behaviours to be checked. We consider the following three examples of delay Painlev\'e equations
\begin{equation} \label{ddP1}
u( \bar{u} - \ubar{u})=  a u - b u',
\end{equation}
\begin{equation} \label{dHKdV}
v^2( \bar{v} - \ubar{v})=  p v + q v',
\end{equation}
\begin{equation} \label{eq20}
\bar{w} w = \ubar{w}\left( \lambda z w + \alpha w' \right),
\end{equation}
where $u, v$ and $w$ are functions of the complex independent variable $z$, we take $p, q, a, b, \lambda, \alpha $ to be complex parameters, and we denote up- and down-shifts by $\bar{u}(z)= u(z+1), \ubar{u}(z) = u(z-1)$ etc.\\
For the purpose of isolating integrability candidates in the class of delay-differential equations, it seems to have been sufficient to require only that the simplest singularities exhibit confinement-type behaviour, and all three of the examples above may be obtained by such means. However, if singularity confinement is to lead to a geometric theory in this case, a more detailed analysis is required. It is the first steps in this direction that we take in this paper, by extending previous observations to account for different multiplicities with which solutions take singular values, as well as giving a geometric description of singularities that may arise in delay-differential equations and what it means for them to be confined. \\

The equation \eqref{ddP1} was obtained by Quispel, Capel and Sahadevan \cite{QCS} as a similarity reduction of the Kac-van Moerbeke differential-difference equation, also known as the Manakov equation or Volterra lattice. They also showed that it has a continuum limit to the first differential Painlev\'e equation and that it exhibits some singularity confinement-type behaviour. The equation \eqref{dHKdV} is a symmetry reduction of a known integrable differential-difference modified Korteveg-de Vries equation, and extensions of it have been studied by Halburd and Korhonen from the point of view of Nevanlinna theory \cite{RODRISTO}. Further, it has a continuum limit to the first Painlev\'e equation and may be obtained from B\"acklund transformations of the third Painlev\'e equation \cite{BJORNTHESIS}, or alternatively using singularity confinement tests adapted from those in \cite{SCLOWDEGREE}. The third equation \eqref{eq20} was isolated as an integrability candidate by Ramani, Grammaticos and Moreira \cite{RGMOREIRA} using a kind of singularity confinement test (which also recovered equation \eqref{ddP1}), and has a continuum limit to the first Painlev\'e equation. We also point out that other integrability properties analogous to those of differential and discrete Painlev\'e equations have been studied in equations \eqref{ddP1},\eqref{dHKdV},\eqref{eq20}, for example the fact that they may be rewritten in bilinear forms \cite{CARSTEABILINEAR} and that degenerate cases admit elliptic function solutions \cite{BJORNTHESIS}, in parallel with the discrete case where autonomous degenerations of discrete Painlev\'e equations are Quispel-Roberts-Thompson (QRT) mappings \cite{QRT1988, QRT1989}, solved by elliptic functions.\\

We also remark that we are considering examples of so-called three-point delay differential equations, which are of the form
\begin{equation} \label{threepoint}
\bar{u} = \frac{ f_1(u,u',...) + f_2(u,u',...) \ubar{u}}{ f_3(u,u',...) + f_4(u,u',...) \ubar{u}},
\end{equation} 
where $f_i$ are polynomials in $u$ and its derivatives. There are known integrable delay-differential equations of other forms, for example the so-called bi-Riccati equations \cite{RGMOREIRA,BJORNRICCATI}, but studies of singularity confinement in these more closely resembles classical Painlev\'e analysis than birational geometry, and will not be discussed in this paper. The class of three-point equations is the one considered by Halburd and Korhonen through the Nevanlinna theoretic approach \cite{RODRISTO}, and fits into the family for which Viallet defined algebraic entropy in the delay-differential setting \cite{DELAYALGENT}.

\subsection{Background}

The differential Painlev\'e equations $\text{P}_{\mathrm{I}}$-$\text{P}_{\mathrm{VI}}$ are six nonlinear second-order ordinary differential equations (ODEs), the study of which has become one of the cornerstones of the field of integrable systems. Painlev\'e, Gambier, Fuchs and their collaborators considered a large class of second-order ODEs, and isolated those for which all solutions are single-valued about any movable singularities (those whose locations depend on the initial conditions). This condition is now known as the Painlev\'e property, and of all the equivalence classes of equations obtained, the six Painlev\'e equations arose as representatives whose general solutions could not be expressed in terms of known functions. These new special functions, known as the Painlev\'e transcendents, play a central role in modern nonlinear physics, see e.g. \cite{CLARKSONSPECIAL,fokas2006painleve} and numerous references within.\\

The differential Painlev\'e equations admit a geometric description in terms of rational surfaces obtained by blowing up certain singularities of the equations. Discovered by K. Okamoto \cite{OKAMOTO1979}, for each equation this comes in the form of a bundle over the independent variable space whose fibres are rational surfaces with certain curves removed. The bundle, known as Okamoto's space, admits a foliation by solution curves of the ODE system transverse to the fibres, and each fibre can be regarded as a space of initial conditions for the system. Further, the curves which were removed from each fibre (the inaccessible divisors) have irreducible components whose intersection configuration is encoded in a Dynkin diagram of affine type, also known as an extended Dynkin diagram. It was also shown that Okamoto's space for each $\text{P}_{\mathrm{I}}-\text{P}_{\mathrm{VI}}$ essentially determines the differential equation \cite{MMT1999, MATSUMIYA1997,TS1997}, and can be used to explain many of their properties (see \cite{KNY2017} and references within). \\

Beginning in the 1990's, important steps were made towards defining and understanding discrete analogues of the Painlev\'e equations, through the proposal by Ramani and Grammaticos, together with Papageorgiou, of singularity confinement \cite{RGP1991} as the discrete counterpart to the Painlev\'e property. We will illustrate the singularity confinement phenomenon in the second order difference equation
\begin{equation} \label{scalarQRT}
f_{n+1} = \frac{(f_n-k)(f_n+k) f_{n-1}}{k^2 - f_n^2 +2 t f_n f_{n-1} },
\end{equation}
with parameters $k \neq 0, \pm 1$ and $t \neq 0$. The initial value problem for this equation requires two values of the solution, say $f_0, f_1$, which in almost all cases will allow the values $f_2, f_3$ and so on to be determined recursively. The system \eqref{scalarQRT} has singular values $f_n = \pm k$, in the sense that if while iterating the solution takes one of these values, $f_{n+1}$ is zero independent of the value of $f_{n-1}$ (provided $f_{n-1} \neq 0$). This is usually referred to as a loss of a degree of freedom occurring while iterating the system. For generic (non-integrable) discrete systems, the singularity propagates, in the sense that the subsequent values $f_{n+2}, f_{n+3}, ...$ will all be determined independently of $f_{n-1}$ and the lost degree of freedom is never recovered. In our case, we may compute the next iterate $f_{n+2} = \mp k$, but then, importantly, arrive at an indeterminacy of the rational function giving $f_{n+3}$, namely at $(f_{n+1}, f_{n+2}) = (0, \mp k)$. If, however, we consider a perturbation of the singular value $f_n = \pm k$ by introducing a small parameter $\varepsilon$, we may compute the following in the small $\varepsilon$ limit:
\begin{equation*} 
f_{n-1} \neq 0, \quad f_{n} = \pm k + \mathcal{O}(\varepsilon), \quad f_{n+1} = \mathcal{O}(\varepsilon), \quad f_{n+2} = \mp k + \mathcal{O}(\varepsilon), \quad f_{n+3} = f_{n-1} + \mathcal{O}(\varepsilon).
\end{equation*}
If we define the values of the iterates as the limits of the above sequence as $\varepsilon \rightarrow 0$, the lost degree of freedom is said to be recovered in the value of $f_{n+3}$, and the singularity at $f_n=\pm k$ is said to be confined. The singularity confinement property for second order discrete systems can be understood as the existence of a space of initial conditions for the system: a family of rational surfaces to which the birational iteration mappings lift to isomorphisms. In fact, defining the values of the solution by iterating and taking limits as above implicitly lifts the system under certain blow-ups.  The example \eqref{scalarQRT} is in fact an example from the family of QRT mappings \cite{QRT1988, QRT1989}, the definition of which ensures they have a space of initial conditions given by a rational elliptic surface. \\

The equation \eqref{scalarQRT} can be considered as a birational mapping $\varphi : \mathbb{P}^1 \times \mathbb{P}^1 \rightarrow \mathbb{P}^1 \times \mathbb{P}^1$. Letting $f_{n-1} = y$, $f_n = x = \bar{y}$, $f_{n+1} = \bar{x}$, the iteration $(f_n, f_{n-1}) \mapsto (f_{n+1}, f_{n})$ gives a birational map $(x,y) \mapsto (\bar{x}, \bar{y})$. We consider this on $\mathbb{P}^1 \times \mathbb{P}^1$ via the usual charts. That is, we use $x,y$ as affine coordinates in the $\mathbb{P}^1$ factors, and introduce $X = 1/x, Y=1/y$, so $\mathbb{P}^1 \times \mathbb{P}^1$ is covered by the four charts $(x,y), (X,y), (x,Y), (X,Y)$. This mapping 
\begin{equation} \label{QRTphi}
\begin{aligned}
\varphi : ~\mathbb{P}^1 \times \mathbb{P}^1 &~\rightarrow~ \mathbb{P}^1 \times \mathbb{P}^1 \\
~~~(x,y) ~&\mapsto ~~(\bar{x}, \bar{y}) = \left( \frac{(x-k)(x+k) y}{k^2 - x^2 +2 t x y },  x \right)
\end{aligned}
\end{equation}
preserves each member of a pencil of elliptic curves on $\mathbb{P}^1 \times \mathbb{P}^1$, and the space of initial conditions is obtained from $\mathbb{P}^1 \times \mathbb{P}^1$ by resolving its basepoints through a number of blow-ups. This is ensured by the definition of the QRT map in terms of this pencil, which we outline now. Consider the matrices
\begin{equation}
\mathbf{A} = \left(
\begin{array}{ccc}
 0 & 0 & \frac{-1}{2 t} \\
 0 & 1 & 0 \\
 \frac{-1}{2 t} & 0 & \frac{k^2}{2 t} \\
\end{array}
\right), \quad \quad \mathbf{B} = \left(
\begin{array}{ccc}
 1 & 0 & 0 \\
 0 & 0 & 0 \\
 0 & 0 & 0 \\
\end{array}
\right),
\end{equation}
where again $k \neq 0, \pm 1$ and $t \neq 0$, which define a pencil of biquadratic curves $\left\{ \Gamma_{[\alpha : \beta ]} ~:~ [\alpha : \beta] \in \mathbb{P}^1 \right\}$ in $\mathbb{P}^1 \times \mathbb{P}^1$, written in the affine coordinates $(x,y)$ as 
\begin{equation} \label{pencil}
\Gamma_{[\alpha : \beta]} : \quad \alpha \mathbf{x}^T \mathbf{A} \mathbf{y} + \beta \mathbf{x}^T \mathbf{B} \mathbf{y} = \frac{\alpha}{2t}(k^2 - x^2 -y^2 +2 t x y ) + \beta x^2 y^2 = 0,
\end{equation}
where $\mathbf{x}^T = \left(\begin{array}{ccc}x^2 & x & 1\end{array}\right), \mathbf{y}^T = \left(\begin{array}{ccc}y^2 & y & 1\end{array}\right)$. The QRT mapping is defined as follows. A generic point, say given by $(x,y)$, lies on exactly one curve $\Gamma_{[\alpha : \beta]}$ in the pencil. There is then exactly one other point $(\bar{x},y)$ on $\Gamma_{[\alpha : \beta]}$ with the same $y$-coordinate, from which we can define the involution $r_x : (x,y) \mapsto (\bar{x},y)$. Similarly we have another involution $r_y : (x,y) \mapsto (x,\bar{y})$, and their composition $r_x \circ r_y$ is the QRT mapping. Following \cite{CDT2017} we introduce the involution $\sigma_{xy} : (x,y) \mapsto (y,x)$ and work with the map $\varphi = \sigma_{xy} \cdot r_y$, which for the pencil \eqref{pencil} is precisely \eqref{QRTphi}, and can be thought of as a `half QRT mapping' due to the fact that $\varphi^2 = r_x \circ r_y$. The pencil \eqref{pencil} has four basepoints, given in coordinates by 
\begin{equation}
p_1 : (x,y) = (k, 0),\quad p_2 : (x,y) = (-k, 0),\quad p_3 : (x,y) = (0, k),\quad p_4 : (x,y) = (0, -k).
\end{equation}
Blowing these up, we denote the blow-up projection by 
\begin{equation*}
\pi_1 : \operatorname{Bl}_{p_1,p_2,p_3,p_4}(\mathbb{P}^1 \times\mathbb{P}^1) \rightarrow \mathbb{P}^1 \times \mathbb{P}^1,
\end{equation*} 
and denote the exceptional curves by $\pi_2^{-1} (p_i) = E_i$ for $i=1,2,3,4$. The proper transform of the pencil under $\pi_2$ still has four basepoints $p_5 \in E_1$, $p_6 \in E_2$, $p_7 \in E_3$, $p_8 \in E_4$, after the blow-ups of which the proper transform of the pencil is basepoint-free and we obtain a rational elliptic surface $\mathcal{X}$. Denote the projection under the second four blow-ups by 
\begin{equation*}
\pi_2 : \mathcal{X} \rightarrow  \operatorname{Bl}_{p_1,p_2,p_3,p_4}(\mathbb{P}^1 \times\mathbb{P}^1),
\end{equation*} 
and the exceptional curves by $\pi_2^{-1} (p_i) = E_i$ for $i=5,6,7,8$. Composing the projections we obtain
\begin{equation*}
\pi = \pi_2 \circ \pi_1: \mathcal{X} \rightarrow \mathbb{P}^1 \times \mathbb{P}^1,
\end{equation*} 
and $\mathcal{X}$ is a rational surface fibred by the proper transform of the pencil. Under $\pi$, we have the preimage of each basepoint $p_1,\dots, p_4$ given by the union of two irreducible curves:
\begin{align*}
&\pi^{-1}(p_1) = (E_1 - E_5) \cup E_5,\quad \quad  \pi^{-1}(p_2)= (E_2 - E_6) \cup E_6, \\
&\pi^{-1}(p_1)= (E_3 - E_7) \cup E_7,\quad \quad \pi^{-1}(p_4)= (E_4 - E_8) \cup E_8,
\end{align*}
where we have used the usual notation for divisors to denote by $E_1-E_5$ the proper transform of $E_1$ under $\pi_2$, and so on, which we illustrate in \autoref{blow-ups}.
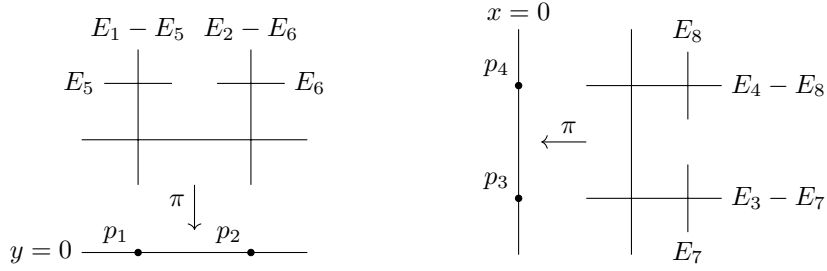
\begin{figure} 
\begin{center}
\begin{tikzpicture}[scale=.3]
\draw (9,-4) node[left]{$y=0$} -- (19, -4) 
(9,1) -- (19,1)
(11.5,-1) -- (11.5, 5) node[above]{$E_1-E_5$}
(16.5,-1) -- (16.5, 5) node[above]{$E_2-E_6$}
(10,3.5) node[left]{$E_5$} -- (13,3.5)
(15,3.5)  -- (18,3.5)node[right]{$E_6$} ;
\draw[->]
(14,-1) node[below left]{$\pi$}-- (14,-3) ;
\filldraw 
(11.5,-4) circle (4pt) node[above left]{$p_1$}
(16.5,-4) circle (4pt) node[above left]{$p_2$};
\end{tikzpicture}
\quad \quad \quad \quad \quad 
\begin{tikzpicture}[scale=.3]
\draw 
(23,-5) -- (23, 5) node[above]{$x=0$}
(28,-5) -- (28, 5) 
(26,-2.5) -- (32,-2.5) node[right]{$E_3-E_7$}
(26,2.5) -- (32,2.5) node[right]{$E_4-E_8$}
(30.5,1) -- (30.5,4) node[above]{$E_8$}
(30.5,-1) -- (30.5,-4) node[below]{$E_7$};
\filldraw 
(23,-2.5) circle (4pt) node[above left]{$p_3$}
(23,2.5) circle (4pt) node[above left]{$p_4$};
\draw[->]
(26,0) node[above left]{$\pi$}-- (24,0) ;
\end{tikzpicture}
\end{center}
\caption{Configuration of curves on $\mathcal{X}$ arising from the blow-ups of the basepoints}
\label{blow-ups}
\end{figure}
The iteration mapping \eqref{QRTphi} lifts uniquely under the blow-ups to give a birational map
\begin{equation*}
\tilde{\varphi} : \mathcal{X} \rightarrow \mathcal{X},
\end{equation*}
which is in fact a true isomorphism, and the singularity confinement observed earlier can be understood in terms of this space of initial conditions as follows. Lifted under the blow-ups, the initial data $f_{n-1} \neq 0, f_n = k$ correspond to a point on the proper transform $H_x - E_1$ of the line $x=k$ on $\mathbb{P}^1 \times \mathbb{P}^1$, while the pairs $(f_{n},f_{n+1})=(k,0)$, $(f_{n+1},f_{n+2})=(0,-k)$ correspond to the basepoints $p_3, p_2$ respectively. Further, the recovery of the degree of freedom $(f_{n+2},f_{n+3}) = (-k, f_{n-1})$ corresponds to a one-to-one correspondence between $H_x- E_1$ and $H_y-E_4$ under the iterated mapping $\tilde{\varphi}^3$, as we illustrate in \autoref{singpattern}. 
\begin{figure}[H] 
\begin{center}
\begin{tikzpicture}[scale=.3]
\draw (1, -1) -- node[midway,left]{$~~~~\mathcal{X}~~~~~~~~~~~$}(1,5) node[above]{$H_x-E_1$}
(8,-1) -- (8,5) node[above]{$E_7$} 
(15,-1) -- (15,5) node[above]{$E_6$} 
(22,-1) -- (22,5) node[above]{$H_y-E_4$} ;
\draw (1, -9) node[below]{$x=k$}-- node[midway, left]{$\mathbb{P}^1\times\mathbb{P}^1~~~~~~~~$} (1,-3)
(22, -9) node[below]{$y=-k$}-- (22,-3);
\filldraw 
(8,-8) circle (4pt) node[below]{$p_3$}
(15,-8) circle (4pt) node[below]{$p_2$};
\draw[->]
(3,2) -- node[above,midway]{$\tilde{\varphi}$} (6,2);
\draw[->]
(10,2) -- node[above,midway]{$\tilde{\varphi}$} (13,2);
\draw[->]
(17,2) -- node[above,midway]{$\tilde{\varphi}$} (20,2);

\draw[->]
(3,-8) -- node[above,midway]{$$} (6,-8);
\draw[->]
(10,-8) -- node[above,midway]{$$} (13,-8);
\draw[->]
(17,-8) -- node[above,midway]{$$} (20,-8);

\draw[->]
(8,-3) -- node[left,midway]{$\pi$} (8,-5);
\draw[->]
(15,-3) -- node[left,midway]{$\pi$} (15,-5);
\end{tikzpicture}
\caption{Confined singularity pattern as isomorphisms between exceptional curves}
\label{singpattern}
\end{center}
\end{figure}
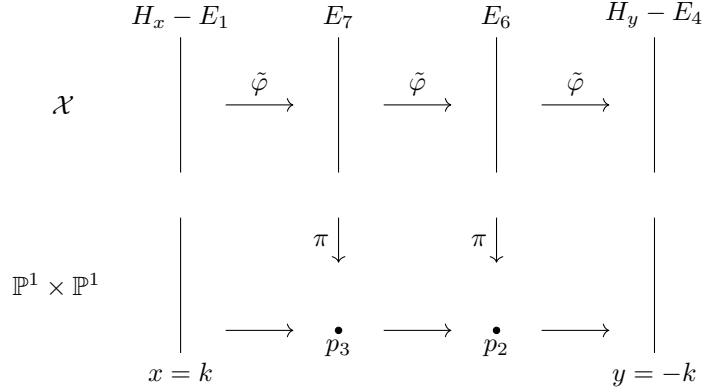
The loss of a degree of freedom when $f_n=\pm k$ can now be understood in terms of curves on $\mathbb{P}^1 \times \mathbb{P}^1$ being blown down to points under the mapping $\varphi$: a codimension one subvariety being blown down to one of codimension two. The recovery of the lost degree of freedom occurs precisely when, while iterating after a blow-down, we arrive at an indeterminacy of the forward iteration map $\varphi$ (in the case of the singularity $f_n = k$, this is $p_2$), so the point is blown back up to a curve. As remarked before, for a generic (non-integrable) system, after a blow-down we will not arrive at an indeterminacy of the forward mapping and the lost degree of freedom will never be recovered. In other words, we cannot lift the mapping to an isomorphism through a finite number of blow-ups. This description of singularity confinement in terms of codimension increasing under the mapping, followed by a return to the same as the generic case, is the main reference point for our geometric formulation of singularity confinement for delay-differential equations in \autoref{section3}. \\ 

Ramani, Grammaticos and collaborators have obtained a plethora of discrete Painlev\'e equations via the process of `deautonomisation by singularity confinement' applied to members of the QRT family. This involves considering non-autonomous generalisations of a given QRT map by introducing $n$-dependence into the coefficients of the mapping, then isolating examples for which the singularity confinement behaviour persists. The definitive framework for discrete Painlev\'e equations was provided in a seminal paper by H. Sakai \cite{SAKAI2001}. Sakai defined a class of rational surfaces generalising both those associated with differential Painlev\'e equations via Okamoto's space and the rational elliptic surfaces giving spaces of initial conditions for QRT mappings. Certain surfaces from this class come in families that admit actions of extended affine Weyl groups by birational transformations, with translation elements defining discrete Painlev\'e equations. The theory of Sakai has had a huge impact on both the general theory of discrete integrable systems, as well as on the applications in which they arise. While this theory provides a classification scheme for discrete Painlev\'e equations in terms of the surfaces they are associated with, it has also led to a suite of geometric tools for their analysis (see \cite{KNY2017} and numerous references within), which are invaluable in cases where a discrete system from an applied problem fits into the discrete Painlev\'e framework \cite{hypergeometric}. Sakai's construction recovers many of the examples obtained by singularity confinement methods, but we make an important remark here that lifting to isomorphisms under a finite number of blow-ups is not sufficient for integrability, and the geometry of the space of initial conditions plays a defining role. In particular, an example given by Hietarinta and Viallet \cite{HIETARINTAVIALLET} admits a space of initial conditions but exhibits exponential degree growth, which was explained in terms of its geometry by Takenawa \cite{TAKENAWA2001}. It has since been shown \cite{MASE} that if a second-order discrete system with the singularity confinement property (in the sense that it admits a space of initial conditions) is nontrivially integrable (in the sense of quadratic degree growth), then it must arise from the surfaces defined by Sakai. \\

As mentioned previously, the theory of delay-differential Painlev\'e equations is in its infancy compared to the differential and discrete cases, but there is already a body of evidence showing its promise, which we hope to add to with this work. Delay-differential equations of the kind we consider arise in a range of fields of applied mathematics, most notably in mathematical biology, for example as equations for steady states of systems of partial differential equations with a spatial delay \cite{HUGH}. Thus the possibility of a geometric framework for Painlev\'e equations in the delay-differential class is an exciting prospect not only for the theory of Painlev\'e equations itself, but for widening the range of equations whose integrability can be exploited in applications. 

\subsection{Outline of the paper}

We will begin our analysis working on the level of equations, without invoking geometric language. In \autoref{section2} we recall previous observations of singularity confinement behaviour in the three equations, and extend them to include infinite families of confined singularity patterns in each case. The proofs of these are deferred to the appendix. In section 3 we shift to the geometric setting, first recasting our equations as mappings between jet spaces and defining `blow-down type' singularities, and propose a notion of confinement for them. Rephrased in these geometric terms, we use the results of Section 2 to show that in the three examples, all such singularities are, in the sense of our definition, confined. We conclude with a discussion of how the geometric framework and the techniques developed for proving the singularity confinement property may be utilised and built upon in the study of other examples, as well as some open questions that arise from our work.

\section{Singularity analysis of delay-differential equations} \label{section2}

We begin by recalling previous observations of singularity confinement phenomena in the three examples we consider. Beginning with equation \eqref{dHKdV}, the forward iteration, which gives $\bar{v}$ in terms of $v, v'$ and $\ubar{v}$ is given by
\begin{equation} \label{vbar}
\bar{v} = \ubar{v} + p\frac{1}{v} + q \frac{v'}{v^2},
\end{equation}
so if we take, as initial data, a pair of Laurent series expansions of $v, \ubar{v}$ about $z=z_0$, then \eqref{vbar} and its upshifts determine all subsequent iterates $\bar{v}, \bar{\bar{v}}, \dots$ as Laurent series about $z_0$. If we only wish to iterate a finite number of steps forward from generic initial data, we need only finitely many coefficients. For example, we could begin by giving initial $\ubar{v}, v$ as Taylor expansions in $\zeta = z- z_0$ about some $z=z_0$:
\begin{subequations} \label{jetsv}
\begin{align}
\ubar{v} &= \ubar{a}_0 + \ubar{a}_1 \zeta + \ubar{a}_2 \zeta^2+  \dots , \label{uunderexpansion} \\
v &= a_0 + a_1 \zeta + a_2 \zeta^2 + \dots.  \label{uexpansion}
\end{align}
\end{subequations}
If we assume that the iterates $\bar{v}(z) = u(z+1), \bar{\bar{v}}(z) = v(z+2), \dots, v^{(k)}(z) = v(z + k)$ are all regular and nonzero at $z_0$, it is clear from the form of the equation \eqref{vbar} that the value $v(z_0 + k)$ depends only on the following coefficients from the expansions \eqref{uunderexpansion}, \eqref{uexpansion}:
\begin{equation}
\left(\begin{array}{ccccc}\ubar{a}_0 & \ubar{a}_1 & \dots & \ubar{a}_{k-1} & \\a_0 & a_1 & \dots & a_{k-1} & a_k \end{array}\right).
\end{equation}
We will be iterating systems arbitrarily many times forward, so we will use this kind of notation for the iterates, i.e. $v^{(k)}(z) = v(z+k)$, throughout the remainder of the paper. Further, the form of the right-hand side of the forward iteration \eqref{vbar} ensures that if we start from $\ubar{v}, v$ given by Taylor series, the only way that a pole may develop is through some iterate having a zero first. If while iterating, some iterate $v$ develops a zero of order one, say at $\zeta = z-z_0= 0$, with
\begin{subequations}
\begin{align}
\ubar{v} &= \ubar{a}_0 + \ubar{a}_1 \zeta + \dots, \\
v &= a_1 \zeta + a_2 \zeta^2 + \dots,
\end{align}
\end{subequations}
where $a_1 \neq 0$, then we have by direct calculation that
\begin{subequations}
\begin{align}
\bar{v} &= - \frac{q}{a_1} \zeta^{-2} + \mathcal{O}(\zeta^{-1}), \\
\bar{\bar{v}} &= - a_1 \zeta^{1} + \mathcal{O}(\zeta^2), \\
\bar{\bar{\bar{v}}} &= \left(5 \ubar{a}_0+\frac{7 p^2}{q a_1}+\frac{2 p a_2}{a_1^2}-\frac{4 q a_2^2}{a_1^3}+\frac{6 q a_3}{a_1^2}\right) + \mathcal{O}(\zeta).
\end{align}
\end{subequations}
We summarise the observations above by saying that the equation \eqref{dHKdV} admits the singularity pattern
\begin{equation*}
\left( \operatorname{rg} , 0^1, \infty^{2} , 0^1, \operatorname{rg} \right),
\end{equation*}
where $\operatorname{rg}$ indicates a regular iterate with generic coefficients. We note here that this behaviour is exceptional for the following reason. In the computation of $\bar{\bar{v}}$ here, it is natural to expect a zero of order one, as this is what happens generically when $v$ and $\bar{v}$ are of order $\zeta, \zeta^{-2}$ respectively. However, while $\bar{v}, \bar{\bar{v}}$ having orders $\zeta^{-2}, \zeta^1$ respectively would generically lead to $\bar{\bar{\bar{v}}}$ having another pole of order $2$, in this singularity pattern we note that two terms have vanished as $\bar{\bar{\bar{v}}}$ regains regularity. In the language of previous studies of singularity confinement behaviour, the information lost when entering the singularity is recovered in the iterate $\bar{\bar{\bar{v}}}$, in the form of the coefficient $\ubar{a}_0$ from the initial data. Though this behaviour has not, to our knowledge, been reported explicitly, we note that the equation \eqref{dHKdV} may be obtained by singularity confinement tests along the lines of \cite{RGMOREIRA,SCLOWDEGREE}.\\

We next consider equation \eqref{ddP1}, which was first observed in \cite{QCS} to exhibit the following singularity confinement behaviour. The forward iteration is given by
\begin{equation} \label{ddP1fwd}
\bar{u} = \ubar{u} + a - b \frac{u'}{u},
\end{equation}
so again it is clear that the only way that a pole may develop while iterating from formal Taylor series is following a zero. Suppose that while iterating, the solution $u$ develops a zero of order one at $\zeta = z-z_0 = 0$, so
\begin{subequations}
\begin{align}
\ubar{u} &= \ubar{c}_0 + \ubar{c}_1 \zeta + \dots, \\
u &= c_1 \zeta + c_2 \zeta^2 + \dots,
\end{align}
\end{subequations}
where $c_1 \neq 0$. Then direct calculation shows that 
\begin{subequations}
\begin{align}
u^{(1)} &= - \frac{b}{\zeta} + \left(a + \ubar{c}_0 - b \frac{c_2}{c_1} \right) + \mathcal{O}(\zeta), \\
u^{(2)} &=  \frac{b}{\zeta} +  \left(2a + \ubar{c}_0 - b \frac{c_2}{c_1} \right) + \mathcal{O}(\zeta), \\
u^{(3)} &= \left(\frac{2 a^2}{b} -\frac{\ubar{c}_0^2}{b}+\frac{2 c_2 \ubar{c}_0}{c_1}-3 \ubar{c}_1+\frac{2 b \left(3 c_1 c_3-2 c_2^2\right)}{c_1^2}-2 c_1 \right)\zeta  + \mathcal{O}(\zeta^2), \\
u^{(4)} &= F(\ubar{c}_0, \ubar{c}_1, \ubar{c}_2, c_1, c_2, c_3, c_4 ) + \mathcal{O}(\zeta),
\end{align}
\end{subequations}
where $F$ is a known rational function of the generic initial data, which we omit for conciseness. Again, this behaviour is exceptional as $u^{(1)}, u^{(2)}$ both having simple poles would generically lead to $u^{(3)}$ also having a simple pole, but here two terms have vanished as $u^{(3)}$ instead has a zero of order one, so equation \eqref{ddP1} admits the singularity pattern 
\begin{equation*}
\left( \operatorname{rg}, 0^{1} , \infty^{1}, \infty^{1}, 0^{1}, \operatorname{rg} \right).
\end{equation*}
We next turn to equation \eqref{eq20}, which was obtained in \cite{RGMOREIRA} by singularity confinement tests, though details were not given explicitly. The forward iteration mapping is given by 
\begin{equation} \label{eq20fwd}
\bar{w} =   \ubar{w} \left( \lambda z  + \alpha \frac{ w'}{w}\right).
\end{equation}
Say, while iterating, we arrive at a pair $\ubar{\ubar{w}}, \ubar{w}$ given by expansions in $\zeta = z-z_0$ by 
\begin{subequations}
\begin{align}
\ubar{\ubar{w}} &= \ubar{\ubar{c}}_0 + \ubar{\ubar{c}}_1 \zeta +\ubar{\ubar{c}}_2 \zeta^2+ \dots, \\
\ubar{w} &= \ubar{c}_0 + \ubar{c}_1 \zeta +\ubar{c}_2 \zeta^2 + \dots,
\end{align}
\end{subequations}
with 
\begin{equation}
\alpha \ubar{c}_1+\lambda(z_0 - 1) \ubar{c}_0 = 0, \quad 2 \alpha  \ubar{c}_2 + \lambda  \ubar{c}_1 (z_0-1) \neq 0, \quad\ubar{c}_1 \neq 0, \quad \ubar{\ubar{c}}_0 \neq 0.
\end{equation}
This means that $w$ will have a simple zero at $z=z_0$, and by direct calculation we find the following:
\begin{subequations}
\begin{align}
w &= \frac{\lambda  \ubar{\ubar{c}}_0 (1-z_0)  \left(2 \alpha  \ubar{c}_2 + \lambda  \ubar{c}_1 (z_0-1)\right)}{\alpha  \ubar{c}_1} \zeta + \mathcal{O}(\zeta^2), \\
\bar{w} &= \frac{\alpha^2 \ubar{c}_1}{\lambda (1-z_0)} \zeta^{-1} + \mathcal{O}(\zeta^0)\\
\bar{\bar{w}} &=  \frac{\lambda  \ubar{\ubar{c}}_0 (z_0-1)  \left(2 \alpha  \ubar{c}_2 + \lambda  \ubar{c}_1 (z_0-1)\right)}{  \ubar{c}_1} + \mathcal{O}(\zeta^1), \\
\bar{\bar{\bar{w}}} &= \frac{ G(\ubar{\ubar{c}}_0,\ubar{\ubar{c}}_1, \ubar{\ubar{c}}_2, \ubar{c}_1,  \ubar{c}_2, \ubar{c}_3) }{\ubar{\ubar{c}}_0^2 \left(2 \alpha  \ubar{c}_2 + \lambda  \ubar{c}_1 (z_0-1) \right)^2}+ \mathcal{O}(\zeta^1),
\end{align}
\end{subequations}
where $G$ is a polynomial function of the generic initial data as well as $z_0$. Again, this behaviour is exceptional as a simple pole of $\bar{w}$ with $\bar{\bar{w}}$ regular and nonzero would generically lead to $\bar{\bar{\bar{w}}}$ having another simple pole, whereas in this case a term has vanished and the iterate $\bar{\bar{\bar{w}}}$ is regular. Again, we summarise this observation by saying that the equation \eqref{eq20} admits the singularity pattern
\begin{equation*}
\left( \operatorname{rg}, \ubar{\zeta}_0^1,  0^{1} , \infty^{1}, \bar{\bar{\zeta}}_0^1 , \operatorname{rg} \right),
\end{equation*}
where $\ubar{\zeta}_0^{(1)}$ indicates that the iterate $\ubar{w}$ satisfies the condition for $w$ to develop a simple zero, namely $\alpha \ubar{c}_1+\lambda(z_0 - 1) \ubar{c}_0 = 0, \quad 2 \alpha  \ubar{c}_2 + \lambda  \ubar{c}_1 (z_0-1) \neq 0$, and $\bar{\bar{\zeta}}_0^1$ indicates the iterate $\bar{\bar{w}} = \bar{\bar{c}}_0 +  \bar{\bar{c}}_1 \zeta +  \bar{\bar{c}}_2 \zeta^2 + \dots $ satisfies $ \alpha  \bar{\bar{c}}_1 + \lambda(z_0+2) \bar{\bar{c}}_0 = 0$.  

\subsection{Infinite families of singularity patterns} \label{infinitefamilies}

In the previous section, we outlined certain singularity patterns admitted by the equations \eqref{ddP1}, \eqref{dHKdV} and \eqref{eq20} which involved zeroes of order one developing while iterating the systems. We now extend these observations to higher order zeroes, and show that each of the equations admits an infinite family of singularity patterns with similar confinement behaviour.\\

For equation \eqref{dHKdV}, we have observed the singularity pattern $\left( \operatorname{rg} , 0^1, \infty^{-2} , 0^1, \operatorname{rg} \right)$, which corresponds to $\underline{v}$ being regular and $v$ having a zero of order one at $z = z_0$. Similarly, if $v$ has a zero of order two, then we pass through the following sequence of orders, which is generic until three terms vanish as $v^{(5)}$ becomes regular instead of a pole (with leading coefficient depending on data from $\ubar{v}$):
\begin{equation*}
\ubar{v} = \mathcal{O}(\zeta^0), \quad v \sim \zeta^2, \quad v^{(1)} \sim \zeta^{-3}, \quad v^{(2)} \sim \zeta^{2}, \quad v^{(3)} \sim \zeta^{-3}, \quad v^{(4)} \sim \zeta^{2}, \quad v^{(5)} = \mathcal{O}(\zeta^{0}).
\end{equation*}
From above, we see that equation \eqref{dHKdV} admits the singularity pattern
\begin{equation*}
\left( \operatorname{rg} , 0^2, \infty^{3} , 0^2, \infty^{3}, 0^2,  \operatorname{rg} \right),
\end{equation*}
and because of the return to regularity and the iterate $v^{(5)}$ depending on the generic initial data from $\ubar{v}$, the singularity is confined in a similar sense to that which we observed in the case of a zero of order one. More generally, if $v$ has a zero of order $m>1$, and $\ubar{v}$ is regular, say $v = c_m \zeta^m + \mathcal{O}(\zeta^{m+1})$, with $c_m \neq 0$, and $\ubar{v} = \mathcal{O}(1)$, then it can be seen from the equation \eqref{dHKdV} that
\begin{subequations} 
\begin{align}
v^{(1)} &=  -\frac{m q}{ c_m} \zeta^{-m-1} + \mathcal{O}(\zeta^{-m}), \\
v^{(2)} &= - \frac{c_m}{m} \zeta^{m} + \mathcal{O}(\zeta^{m+1}), \\
v^{(3)} &=  \frac{m(m-1) q}{c_m} \zeta^{-m-1} + \mathcal{O}(\zeta^{-m}),
\end{align}
\end{subequations} 
and more generally, it can be shown by induction that for $k \leq m$, 
\begin{subequations}\label{uiterates}
\begin{align} 
v^{(2k)} &=\frac{(-1)^k k! }{\prod_{i=0}^{k-1} (m-i)} c_m \zeta^{m} + \mathcal{O}(\zeta^{m+1}), \\
v^{(2k+1)} &= \frac{(-1)^k \prod_{i=0}^{k}(m-i)}{k!} \frac{q}{c_m} \zeta^{-m-1} + \mathcal{O}(\zeta^{-m}).
\end{align}
\end{subequations}
What we deduce from this is that a singularity sequence beginning with $\ubar{v}$ regular and $v$ with a zero order $m$ will contain a sequence of $m+1$ zeroes of order $m$ alternating with $m$ poles of order $m+1$. We know that the coefficient of $\zeta^{-m-1}$ in the iterate $v^{(2m+1)}$ will vanish according to the formulae \eqref{uiterates}, but it turns out that the entire singular part of the expansion vanishes, so regularity is regained at the iterate $v^{(2m+1)}$.

\begin{theorem}\label{theoremeq2}
For each integer $m>0$, equation \eqref{dHKdV} admits the singularity pattern
\begin{equation}
\left( \operatorname{rg}, 0^m , \infty^{m+1}, 0^m, \infty^{m+1}, \dots,  \infty^{m+1},0^m ,  \infty^{m+1}, 0^m, \operatorname{rg} \right),
\end{equation}
which includes $m+1$ zeroes of order $m$ alternating with $m$ poles of order $m+1$. 
\end{theorem}

The proof of this theorem is provided in the appendix, along with those of similar results for the equations \eqref{ddP1} and \eqref{eq20}:
\begin{theorem} \label{theoremeq1}
For each integer $m > 0$, equation \eqref{ddP1} admits the singularity pattern
\begin{equation}
\left( \operatorname{rg}, 0^m , \infty^1_{-m}, \infty^1_{1}, \infty^{1}_{1-m} , \dots, \infty^1_{k}, \infty^1_{k-m} ,\dots , \infty^1_{-1} , \infty^1_{m}, 0^m, \operatorname{rg} \right),
\end{equation}
which includes $2m$ simple poles with residues alternating between positive and negative multiples of $\beta$, which we denote 
\begin{equation}
\infty^1_j = \frac{j \beta}{z-z_0} + O(1).
\end{equation}
\end{theorem}

\begin{theorem}\label{theoremeq3}
Equation \eqref{eq20} admits the singularity pattern
\begin{equation}
\left( \operatorname{rg},  \zeta_0(-1)^m , 0^m , \infty^{1}, 0^{m-1}, \infty^{2}, \dots, \infty^{j}, 0^{m-j} ,\dots , 0^{2} ,  \infty^{m-1},0^1 ,  \infty^{m},  \zeta_0(2m+2)^m, \operatorname{rg} \right),
\end{equation}
where $\zeta_0(-1)^m$ indicates that the iterate $\ubar{w} = w^{(-1)}$ satisfies $\frac{d^k}{dz^k} \left(\lambda z \ubar{w} + \alpha \ubar{w}' \right) = 0$ at $z=z_0$ for $k=0, ..., m-1$, and $\zeta_0(2m+2)^m$ indicates that the iterate $w^{(2m+2)}$ satisfies $$\frac{d^k}{dz^k} \left(\lambda (z+m) w^{(2m+2)} + \alpha' w^{(2m+2)} \right) = 0,$$ at $z=z_0$ for $k=0, ..., m-1$.
\end{theorem}

\section{Geometric description of singularity confinement} \label{section3}

We now rephrase the results of the previous section geometrically, and propose a characterisation of singularity confinement in the delay-differential setting in terms of the birational geometry of jet spaces. Our guiding principle in developing the theory in parallel with the discrete setting will be that of generic information loss, in particular the ways in which iterating a delay-differential equation may result in a departure from this, and in what sense it is recovered. To explain the motivations for this analogy, we first note that a birational mapping between smooth projective algebraic surfaces is an isomorphism between Zariski open subsets given by the complement of proper subvarieties that are blown down by either the mapping or its inverse. Almost all curves are mapped bijectively to curves, and in this sense no information loss occurs generically while iterating the corresponding discrete system. Singularities of a second-order discrete system occuring when curves are blown down to points may be interpreted as more information loss occurring than normal. The system having the singularity confinement property means that, in such a case when iterating the system results in more than the generic amount of information loss, we may compose the mapping a finite number of times to recover the generic behaviour: an isomorphism from a curve to a curve.\\

We will formulate a concept of generic information loss for our delay-differential equations. In terms of this we will define singularity confinement as being able to, in the case when iterating the system results in more than generic levels of information loss, compose the iteration mapping of the system a finite number of times to recover the generic amount. This concept of generic information loss has two elements: First is the amount of initial data required generically to iterate the system forward a given number of times, which we will phrase in \autoref{jetspaces} in terms of the orders of jet spaces on which the systems give well-defined mappings. Second is the behaviour of subspaces under the these mappings in terms of their codimension, which will be used to describe phenomena analogous to degrees of freedom being lost, which we define as `blow-down type' singularities in \autoref{blow-downsingularities}. We then outline what it means for such a singularity to be confined, and finally verify that this geometric description fits with our analysis of the three examples, and that they confine all singularities in this sense.

\subsection{Delay-differential equations as mapping between jet spaces} \label{jetspaces}

Similarly to how second-order discrete systems are described by birational mappings between algebraic surfaces, we will recast our delay-differential equations as mappings between jet spaces. We consider jets associated with the trivial bundle over $\mathbb{C}$ with fibre $\mathbb{P}^1 \times \mathbb{P}^1$. We use the same coordinate charts for $\mathbb{P}^1 \times \mathbb{P}^1$ as in the discrete case, namely $(x,y), (X,y), (x,Y), (X,Y)$ where $X=1/x, Y=1/y$. The space $J^r_{z_0}$ of $r$-jets about $z_0$ is the set of equivalence classes of local holomorphic sections about some $z_0 \in \mathbb{C}$ under the following equivalence relation. The sections $\sigma_1, \sigma_2$ define the same $r$-jet if, when written in coordinates, their derivatives at $z_0$ coincide up to and including order $r$. \\

We will be always considering jets at $z_0$, so we omit the subscript. We will use coordinates for $J^r$ induced by writing sections as expansions in our coordinates for $\mathbb{P}^1 \times \mathbb{P}^1$. For example, if a section about $z_0$ is visible in the $(x,y)$-chart, it may be written in coordinates as
\begin{equation} \label{section}
\left(\begin{array}{c}x(z) \\ y(z) \end{array}\right) = \left(\begin{array}{c}x_0 + x_1 \zeta + x_2 \zeta^2 + \dots   \\ y_0 + y_1 \zeta + y_2 \zeta^2 + \dots  \end{array}\right),
\end{equation}
where $\zeta = z-z_0$ as before, so we have one part of $J^r$ covered by the chart with coordinates
\begin{equation}
\left(\begin{array}{ccccc}x_0 & x_1 & x_2 & \dots & x_r \\y_0 & y_1 & y_2 & \dots & y_r\end{array}\right),
\end{equation}
and $J^r$ can be thought of as four copies of $\mathbb{C}^{2r+2}$ with coordinates being coefficients from expansions of sections in the  four charts for $\mathbb{P}^1 \times \mathbb{P}^1$, with gluing determined by that of $\mathbb{P}^1 \times \mathbb{P}^1$ itself, namely $X = 1/x, Y=1/y$.\\

Consider a three-point delay-differential equation of the form \eqref{threepoint} given in the introduction, with $l$ being the highest order of derivative that appears. Similarly to how the scalar difference equation \eqref{scalarQRT} is  recast as a QRT mapping on $\mathbb{P}^1 \times \mathbb{P}^1$, we let $(x, y) = (u, \ubar{u})$ and $(\bar{x}, \bar{y}) = (\bar{u}, u)$ given by series expansions about $z_0$, so we have a mapping on sections near $z_0$, which in the $(x,y)$ charts for both domain and target copies of $\mathbb{P}^1 \times \mathbb{P}^1$ is written as:
\begin{equation} \label{sectionmap}
\begin{aligned}
&\quad \quad \quad  \quad \quad \quad \left(\begin{array}{c}x(z) \\ y(z) \end{array}\right) \mapsto \left(\begin{array}{c}\bar{x}(z) \\ \bar{y}(z) \end{array}\right), \\
\bar{x} = & \frac{ f_1 (x, x', \dots, \partial^l x / \partial z^l )+f_2 (x, x', \dots ,\partial^l x / \partial z^l  ) y}{f_3 (x, x', \dots ,\partial^l x / \partial z^l  )+f_4 (x, x', \dots , \partial^l x / \partial z^l  ) y} , \quad \quad \bar{y}= x.
\end{aligned}
\end{equation}

We now introduce a space of jets on which we consider this, corresponding to generic initial data. Consider a section written as a series expansion in one of the four coordinate charts for $\mathbb{P}^1 \times \mathbb{P}^1$, for example \eqref{section} in the $(x,y)$ chart. Denote the numerator and denominator of the function giving $\bar{x}(z)$ in this chart by $P(z), Q(z)$, so for example in the $(x,y)$ chart we use \eqref{sectionmap} and consider
\begin{equation}
\begin{aligned}
P &= f_1 (x, x', \dots, \partial^l x / \partial z^l )+f_2 (x, x', \dots ,\partial^l x / \partial z^l  ) y, \\
Q &= f_3 (x, x', \dots ,\partial^l x / \partial z^l  )+f_4 (x, x', \dots , \partial^l x / \partial z^l  ) y.
\end{aligned}
\end{equation}
Substitute expansions giving  $(x(z), y(z))$ into these, to obtain formal expansions of $P(z), Q(z)$ about $z_0$, which we denote
\begin{equation}
P(z) = P_0 + P_1 \zeta + P_2 \zeta^2+ \dots, \quad \quad Q(z) = Q_0 + Q_1\zeta + Q_2 \zeta^2 + \dots,
\end{equation}
where $P_0, Q_0$ are polynomials in $x_0, \dots, x_l, y_0$ because of the highest order derivative appearing in the equation (or the equivalent for an expansion of a section in another coordinate chart). Consider the rational function $P_0/Q_0$ on $J^{r+l}$, using the transition functions between $x_i, X_i$ etc. being defined by the $\mathbb{P}^1 \times \mathbb{P}^1$ gluing as before, and denote its indeterminacy locus (where the numerator and denominator simultaneously vanish) by $I_1$. We then have a well-defined map 
\begin{equation}
\varphi_r : J^{r+l} \backslash I_1 \rightarrow J^{r}.
\end{equation}
The reason we do not have to worry about indeterminacies of rational functions giving later coefficients in the expansion of $P/Q$ to obtain a well-defined map is the following: All of the rational functions giving expansions of $P/Q$ have denominator being a power of $Q_0$. Similarly, all rational functions giving coefficients in the expansion of $Q/P$ are powers of $P_0$. Thus if $Q_0 = 0$ but $P_0 \neq 0$, we get a well-defined expansion of $Q/P$, in which none of the coefficients have indeterminacies (their denominators cannot vanish as $P_0 \neq 0$) so we have a well-defined a section visible in the $(\bar{X}, \bar{y})$ chart. Similarly, if $P_0 = 0$ but $Q_0 \neq 0$, we get a well-defined expansion of $Q/P$, in which none of the coefficients have indeterminacies (their denominators cannot vanish, as $P_0 \neq 0$). 

\begin{example} 
If we consider the mapping induced by equation \eqref{ddP1} applied to a section visible in the $(x,y)$ chart, written as an expansion \eqref{section}, direct substitution yields
\begin{subequations}
\begin{align}
\bar{x}_0 &= \frac{a x_0 - b x_1 + x_0 y_0}{x_0}, &&\bar{x}_1= \frac{b x_1^2 - 2 x_0 x_2 + x_0^2 y_1}{x_0^2}, &&&\dots \\
\bar{y}_0 &= x_0, &&\bar{y}_1 = x_1, &&&\dots
\end{align}
\end{subequations}
so when $x_0 \neq 0$ we have a section visible in the $(\bar{x},\bar{y})$ chart for the target bundle. Similarly, if we have a section written in the $(X,Y)$ chart as an expansion with coefficients $X_i, Y_i$, we may use the chart $(\bar{x},\bar{Y})$ and calculate
\begin{subequations}
\begin{align}
\bar{x}_0 &=\frac{a X_0 Y_0+b X_1 Y_0+X_0}{X_0 Y_0}, &&\bar{x}_1=\frac{2 b X_2 X_0 Y_0^2-b X_1^2 Y_0^2 - X_0^2 Y_1}{X_0^2 Y_0^2}, &&&\dots \\
\bar{Y}_0 &= X_0, &&\bar{Y}_1 = X_1, &&&\dots
\end{align}
\end{subequations}
so when $X_0 Y_0 \neq 0$ we have a section visible in the $(\bar{x},\bar{Y})$ chart for the target bundle. Calculating in the other charts, we find the subset $I_1 \subset J^{r+1}$ is defined by
\begin{equation}
I_1 = \left\{ (x_0 , x_1) = (0,0) \right\} \cup \left\{ (X_0, X_1) = (0,0) \right\} \cup \left\{ (x_0, Y_0) = (0,0) \right\} \cup \left\{ (X_0, Y_0) = (0,0) \right\}.
\end{equation}
So we have, for each $r \geq 0$, a map
\begin{equation}
\varphi_r : J^{r+1} \backslash I_1 \rightarrow J^{r}.
\end{equation}
We note that the domain $J^{r+1}$ corresponds to the lowest order of jets to which the equation \eqref{ddP1} gives a well-defined map from $J^{r+1} \backslash I_1$ to $J^r$. 
\end{example}
Returning to the general case, we also have, for each $r \geq 0$, a map
\begin{equation} \label{phirk}
\varphi_r^{(k)} = \varphi_{r} \circ \varphi_{r+1} \dots \circ  \varphi_{r+k-1}~~:~~ J^{r+k l} \backslash I_k\rightarrow J^{r},
\end{equation}
defined on the Zariski open subset of $J^{(r+kl)}$ where the numerators and denominators of the rational functions giving leading coefficients of successive iterates do not simultaneously vanish. 
\begin{example}
To illustrate this, in the case of equation \eqref{ddP1} being iterated twice, we obtain in the $(x,y)$ chart rational functions giving $(\bar{\bar{x}}_0, \bar{\bar{y}}_0)$ as 
\begin{subequations} 
\begin{align}
\bar{\bar{x}}_0 &= \frac{a^2 x_0^2-a b x_1 x_0+a x_0^2 y_0+a x_0^3+2 b^2 x_2 x_0-b^2 x_1^2-b x_0^2 y_1-b x_1 x_0^2+x_0^3 y_0}{x_0 \left(a x_0-b x_1+x_0 y_0\right)}, \\
\bar{\bar{y}}_0 &= \frac{a x_0 - b x_1 + x_0 y_0}{x_0}.
\end{align}
\end{subequations} 
Computing the indeterminacy loci of these rational functions in all charts and taking its union with $I_1$, we obtain
\begin{equation}
\begin{aligned}
I_2 &= \left\{ (x_0 , x_1) = (0,0) \right\} \cup \left\{ (X_0, X_1) = (0,0) \right\} \cup \left\{ (x_0, Y_0) = (0,0) \right\} \cup \left\{ (X_0, Y_0) = (0,0) \right\} \cup\\
 & \left\{a x_0 - b x_1 + x_0 y_0 = b x_1^2 - 2b x_0 x_2 + x_0^2 y_1 = 0 \right\} \cup \left\{X_0 = 0, X_1 = -1/b \right\} \cup \left\{ Y_0 = Y_1 = 0 \right\},
\end{aligned}
\end{equation}
and we have a well-defined map 
\begin{equation}
\varphi_r^{(2)} = \varphi_{r} \circ \varphi_{r+1}~~:~~ J^{r+2} \backslash I_2 \rightarrow J^{r}.
\end{equation}
\end{example}

We interpret this map $\varphi_r^{(k)}$ in \eqref{phirk}  on the set specified above as the generic behaviour of the system, and in particular the initial data that is required to iterate the system $k$ times in almost all cases. We now consider the parts of the jet spaces where the rational functions we have considered above have indeterminacies. For example, if we consider a jet in the charts coming from $(X,Y), (\bar{X}, \bar{Y})$, if $(X_0, Y_0) = (0,0)$ then we have
\begin{subequations} \label{poles}
\begin{align}
\bar{X}_0 &=0, &&\bar{X}_1= \frac{Y_1}{1+b Y_1},  \quad \bar{X}_2 = \frac{X_1 \left(Y_2-a Y_1^2\right)-b X_2 Y_1^2}{X_1 \left(1+b Y_1\right)^2},  &&&\dots \\
\bar{Y}_0 &= 0, &&\bar{Y}_1 = X_1, \quad \quad \quad ~~\bar{Y}_2 = X_2, &&&\dots
\end{align}
\end{subequations}
and so on. By direct calculation using formal series expansions, it can be seen that as long as $X_1 \neq 0, 1+b Y_1 \neq 0$, the jet in $(\bar{X},\bar{Y})$ coordinates is determined up to the same order as the one in $(X,Y)$ coordinates. Thus, on the part of $J^r (r \geq1)$ where $(X_0, Y_0) = (0,0)$ but $X_1 \neq 0, 1+b Y_1 \neq 0$, the system induces a mapping $J^{r} \rightarrow J^{r}$ and we have less information loss than in the generic case. Comparing this to the discrete case, we see a parallel to the fact that indeterminacies of the iteration mappings are blown up to curves.

\subsection{Blow-down type singularities} \label{blow-downsingularities}

After considering a concept of generic information loss in terms of the amount of initial data generically required to iterate $k$ times, we turn to parts of jet spaces on which the system induces maps with more information loss. We will refer to these as blow-down type singularities, in parallel with the discrete case where information loss corresponds to curves being blown down under iteration mappings. \\

Consider the mapping $\varphi_r : J^{r+1} \backslash I_1 \rightarrow J^{r}$ induced by equation \eqref{ddP1} derived above. We will be interested in the behaviour under this mapping of subvarieties defined locally by a finite number of algebraic constraints. For most codimension $m$ subsets of this part of $J^{r+1}$ (where $r$ is chosen large enough such that it includes all the variables appearing in the constraints defining the subset), the image under $\varphi_r$ will be of codimension $\leq m$ in $J^{r}$. \\

For example, we can see a variety of behaviours of subspaces as follows. The subspace defined in the $(x_i,y_i)$ chart by the single algebraic constraint $y_i = c$, where $i \leq {r+1}$ and $c\neq0$ is some constant, is of codimension one, and its image under $\varphi_r$ is of codimension zero. Another subspace defined by $x_i = c$, for some $i \leq r$ and $c$ again a nonzero constant, will have image under $\varphi_r$ of codimension one. The codimension two subspace where $(X_0,Y_0)=(0,0)$ with the rest of the coefficients $X_i, Y_i$ generic can be quickly seen from \eqref{poles} to have image again of codimension two.
\begin{definition}
A blow-down type singularity of a delay differential equation of the form \eqref{threepoint} is a codimension $m$ subvariety of $J^{r+l}$ , for some $r\geq 0$, (locally defined as the vanishing locus of a number of polynomials in coordinates introduced above) whose image under the induced map $\varphi_r$ is of codimension greater than $m$.
\end{definition}
We emphasise again that this is in analogy with the discrete setting, where singularities are defined in the sense of an increase in codimension, namely where curves are blown down to points under the iteration mappings. Again we note that in the following examples, $r$ is taken large enough such that $J^{r+1}$ includes all variables appearing in the algebraic constraints defining the blow-down singularities.
\begin{example} \label{exampleeq1sings}
The equation \eqref{ddP1} has a blow-down singularity in $J^{r+1} \backslash I_1$ given in coordinates by $x_0 = 0$ which is of codimension one (with all other $x_i, y_i$ generic) but has image of codimension three in $J^r$, given in coordinates as follows:
\begin{align*}
\left\{ x_0 = 0 \right\} \quad &\rightarrow \quad \left\{ \bar{X}_0 = 0, \quad \bar{X}_1 = -1/b, \quad \bar{y}_ 0 = 0 \right\} \\
\operatorname{codim } 1~~~~ &\rightarrow~~~~ \operatorname{codim } 3
\end{align*} 
Similarly, we see that the development of double and triple zeroes correspond to the following blow-down singularities:
\begin{align*}
\left\{ x_0 =0, \quad x_1 = 0 \right\} \quad &\rightarrow \quad \left\{ \bar{X}_0 = 0, \quad \bar{X}_1 = -1/2b, \quad \bar{y}_ 0 =0, \quad \bar{y}_1 = 0 \right\} \\
\operatorname{codim } 2~~~~ &\rightarrow~~~~ \operatorname{codim } 4
\end{align*} 

\begin{align*}
\left\{ x_0 =0, \quad x_1= 0, \quad x_2 = 0 \right\} \quad &\rightarrow \quad \left\{ \bar{X}_0 = 0, \quad \bar{X}_1 = -1/3b, \quad \bar{y}_ 0 =0, \quad \bar{y}_1 = 0,\quad \bar{y}_ 2 =0 \right\} \\
\operatorname{codim } 3~~~~ &\rightarrow~~~~ \operatorname{codim } 5
\end{align*} 
and more generally the development of a zero of order $m$ corresponds to the following blow-down singularity:
\begin{align*}
\left\{ x_i =0,~~ \forall i = 0, \dots, m-1\right\} \quad &\rightarrow \quad \left\{ \bar{X}_0 = 0, \quad \bar{X}_1 = -1/mb, \quad \bar{y}_i = 0,~~  \forall i = 0, \dots, m-1 \right\} \\
\operatorname{codim } m~~~~ &\rightarrow~~~~ \operatorname{codim }~ (m+2)
\end{align*} 
\end{example}

\begin{example}\label{exampleeq2sings}
The equation \eqref{dHKdV} has a blow-down singularity given in coordinates by $x_0 = 0$ which is of codimension one (with all other $x_i, y_i$ generic) but has image of codimension five given in coordinates as follows:
\begin{align*}
\left\{ x_0 = 0 \right\} \quad &\rightarrow \quad \left\{ \bar{X}_0 = 0, \quad \bar{X}_1 = 0, \quad \bar{y}_ 0 = 0, \quad \bar{y}_1 = q \bar{X}_2, \quad p \bar{y}_1 = - q^2 \bar{X}_3 \right\} \\
\operatorname{codim } 1~~~~ &\rightarrow~~~~ \operatorname{codim } 5
\end{align*} 
We also have a blow-down singularity corresponding to the development of a double zero 
\begin{align*}
\left\{ x_0 =0, \quad x_1 = 0 \right\} \quad &\rightarrow \quad  \left\{\begin{array}{c} \bar{X}_0 = 0, \quad \bar{X}_1 = 0, \quad \bar{X}_2 = 0, \quad\bar{y}_0 = 0, \quad \bar{y}_1 = 0\\  \bar{y}_2 - 2q \bar{X}_3 = 0, \quad \bar{y}_3 - 2p \bar{X}_3 - 4 q \bar{X}_4 =  0, \\
p^2 \bar{X}_3^2+2 p q \bar{X}_3 \bar{X}_4+2 q^2 \bar{X}_4^2-2 q^2 \bar{X}_3 \bar{X}_5 = 0 \end{array}\right\} \\
\operatorname{codim } 2~~~~ &\rightarrow~~~~ \operatorname{codim } 8
\end{align*} 
and more generally the development of a zero of order $m$ corresponds to the following blow-down singularity:
\begin{align*}
\left\{ x_i =0,~~ \forall i = 0, \dots, m-1\right\} \quad &\rightarrow \quad  \left\{\begin{array}{c} \bar{X}_i = 0 ~~ \forall i = 0, \dots, m, \quad \bar{y}_i = 0,~~  \forall i = 0, \dots, m-1\\  F_i( \bar{X}_{m+1}, \bar{X}_{m+2} , \dots ,\bar{y}_{m}, \bar{y}_{m+1}, \dots ) = 0 ~~\forall i=m+1, \dots, 2m+1 \end{array}\right\} \\
\operatorname{codim } m~~~~ &\rightarrow~~~~ \operatorname{codim } ~(3m+2)
\end{align*} 
Here $F_i$ are polynomial in their variables that give $m+1$ independent algebraic constraints, which may be identified by substituting series expansions for $x(z), y(z)$ and noting that $\bar{X}_{2m+2}$ is the first coefficient in which any $y_i$ appears.
\end{example}
\begin{example}\label{exampleeq3sings}
The equation \eqref{eq20} has a blow-down singularity in ($\ubar{x}_i, \ubar{y}_j$ coordinates) corresponding to $x(z)$ developing a zero of order one. This is given by
\begin{align*}
\left\{ (z_0 - 1) \lambda \ubar{x}_0 + \alpha \ubar{x}_1 =  0\right\} \quad &\rightarrow \quad \left\{ x_0 = 0, \quad (z_0 - 1) \lambda y_0 + \alpha y_1 = 0  \right\} \\
\operatorname{codim } 1~~~~ &\rightarrow~~~~ \operatorname{codim } 2
\end{align*} 
and more generally the development of a zero of order $m$ corresponds to the following blow-down singularity, which for conciseness we write in terms of derivatives of the sections, as opposed to explicitly in terms of coefficients:
\begin{align*}
\left\{ \begin{array}{c}  \frac{d^i}{dz^i} \left(\lambda z \ubar{x}(z) + \alpha \ubar{x}'(z) \right) |_{z=z_0}= 0\\  \forall i = 0, \dots, m-1 \end{array} \right\} \quad &\rightarrow \quad  \left\{\begin{array}{c} x_i = 0 \quad \forall i=0,\dots, m-1\\ \frac{d^i}{dz^i} |_{z=z_0} \left(\lambda z y(z) + \alpha y'(z) \right) = 0\\  \forall i = 0, \dots, m-1 \end{array}\right\} \\
\operatorname{codim } m~~~~ &\rightarrow~~~~ \operatorname{codim } ~2m
\end{align*} 
\end{example}

\subsection{Singularity confinement in equations (1.1-1.3)}

We now formulate a geometric description of the confinement type behaviour we observed in our three examples. Again, the analogy with the discrete case is that if, when iterating the system, we arrive at a blow-down type singularity we only need to iterate a finite number of times further to recover the generic level of information loss, both in terms of orders of jet spaces between which the system induces maps, and the behaviour of the singularity under these in terms of codimension. 

\begin{definition} \label{defconfinement}
Consider a three-point delay differential equation of the form \eqref{threepoint} with iteration mappings $\varphi_r$, which has a blow-down type singularity $B_m$ of codimension $m$. We say the singularity $B_m$ is confined if there exists some $k >0$ such that iterating the system $k$ times induces a map from $B_m \subset J^{r+kl}$ whose image is of codimension $\leq m$ in $J^{r}$.
\end{definition}

We note that this definition captures both the recovery from the increase in codimension of $B_m$ as well as the amount of initial data required to iterate $k$ times generically. Take $B_m$ as a subset of the same order jet space $J^{r+kl}$ as for the generic behaviour $\varphi^{(k)}_r : J^{r+kl} \backslash I_k \rightarrow J^{r}$. We consider accessible blow-down singularities: those that may arise when iterating the system from regular nonzero initial data. For the three equations we consider, we first describe the set of all such singularities and then use our results concerning infinite families of singularity patterns to deduce that they are all confined in the above sense. 
\subsubsection{Equation (1.1)}
\begin{lemma} \label{lemmaeq1}
The only accessible blow-down type singularities of equation \eqref{ddP1} are 
\begin{equation*}
B_m = \left\{ x_i = 0 ~~\forall i =0,\dots, m-1 \right\}.
\end{equation*}
\end{lemma}
\begin{proof}
We will first show that the only blow-down singularities visible in the $x_i, y_j$ chart are contained in $\left\{x_0 = 0\right\}$. Suppose $B \subset J^{r+1}$ is of codimension $m$, so dimension $d=2(r+1) - m$, defined locally by $F_1 = \dots = F_l = 0$, where $F_i$ are polynomial in $x_0, \dots, x_m, y_0, \dots, y_m$, and that $x_0 \neq 0$ on $B$. Then near $p \in B$ (at which $B$ is nonsingular) given in coordinates by $p : (x_i, y_j) = (x_i^*, y_j^*)$, we have a parametrisation of $B$ by $d$ free parameters. That is, there exist $i_1, \dots, i_p$, $j_1,\dots j_{d-p} \subset \left\{0,\dots, r+1\right\}$ such that we have a parametrisation
\begin{subequations}
\begin{align}
\left(\begin{array}{c} s_1 \\ \vdots  \\ s_p \\ t_1 \\ \vdots  \\ t_{d-p} \end{array}\right) \mapsto \begin{array}{c} x_{i_1} = x_{i_1}^{*} + s_1 \\ \vdots \\ x_{i_p} = x_{i_p}^{*} + s_p  \\ y_{j_1} = y_{j_1}^{*} + t_1 \\ \vdots \\y_{j_{d-p}} = y_{j_{d-p}}^{*} + t_{d-p} \end{array}
\end{align}
\end{subequations}
with the rest of the variables $x_i, y_j$ given by analytic functions of $s_1, \dots, s_p$, $t_1, \dots , t_{d-p}$:
\begin{equation}
x_i = x_i^* + F_i(s_1, \dots, s_p, t_1, \dots, t_{d-p}), \quad \quad y_j = y_j^* + G_j(s_1, \dots, s_p, t_1, \dots, t_{d-p}),  
\end{equation}
for $i \not\in \left\{ i_1, \dots, i_p \right\}$, $j \not\in \left\{ j_1, \dots, j_{d-p} \right\}$, with $F_i, G_j$ anaytic and zero when all $s_i, t_j$ are zero, and the Jacobian of this parametrisation at $p$ is of rank $d$. We now show, using this parametrisation, that the image of $B$ in $J^r$ under $\varphi_r$ is of dimension $\geq 2r-m$ as long as $x_0 \neq 0$ on $B$. In coordinates, the mapping is of the form
\begin{equation}
\bar{y}_n = x_n, \quad \quad \bar{x}_n = y_n - \frac{P_n (x_0 ,\dots, x_{n+1})}{x_0^{n+1}}.
\end{equation}
Here $P_n$ is a homogeneous polynomial of degree $n+1$, which follows from the repeated application of the quotient rule in computing expressions for derivatives of $\bar{x} = y + \frac{ax-bx'}{x}$. We obtain a local parametrisation of the image of $B$:
\begin{equation}
 \begin{array}{l}\bar{y}_{i_1} = x_{i_1}^* + s_1 \\ \quad \quad \vdots \\ \bar{y}_{i_p} = x_{i_p}^* + s_p \end{array} \quad   \begin{array}{l} \bar{x}_{j_1} = y_{j_1}^* + t_1 + H_1 \\ \quad \quad  \vdots \\\bar{x}_{j_{d-p}} = y_{j_{d-p}}^* + t_{d-p} + H_{d-p} \end{array}
\end{equation}
where $H_1, \dots H_{d-p}$ are analytic in $s_1, \dots, s_p$ (as $x_0 \neq 0$ on $B$), with the rest of the coordinates $\bar{y}_i, \bar{x}_j$ being analytic functions of the parameters. The Jacobian of this parametrisation can be seen to have rank at least $d-2$, with linearly independent columns corresponding to partial derivatives with respect to $s_1, \dots s_p, t_1, \dots, t_{d-p-1}$ ($t_{d-p}$ will not contribute to the rank if $d-p=r+1$, i.e. if $y_{r+1}$ is one of the free variables in the parametrisation of $B$). The possibility that the image is of codimension less than $m$ has already been illustrated at the start of subsection 3.2, where constraints on $y_j$ may not induce constraints on the image. \\

Similarly, if we consider a subvariety of codimension $m$ in the chart $(X,y)$ away from $\{X_0=0\}$, we see that its image under $\varphi_r$ must be again of codimension $\leq m$. This is done in exactly the same way as above, noting that the mapping in charts is of the form 
\begin{equation}
\bar{Y}_n = X_n, \quad \quad \bar{x}_n = y_n - \frac{ P_n (X_0 ,\dots, X_{n+1})}{X_0^{n+1}},
\end{equation}
where again $P_n$ is a homogeneous polynomial of degree $n+1$. Regarding the part of the jet space with $X_0 = 0$, we remark that $X_0 = 0$ with $y_0 \neq0$ is not an accessible singularity, as for a pole to develop while iterating, it must follow a zero. Further, the only parts of $\left\{X_0 = 0, y_0 = 0\right\}$ accessible from regular and nonzero initial data are those coming from one of the blow-down singularities $B_m$. Similar calculations in the charts $(x,Y)$ and $(X,Y)$ show that it suffices to consider blow-down singularities visible in the $(x,y)$ chart where at least $x_0 = 0$. If we take $x(z) = x_m \zeta^m + x_{m+1} \zeta^{m+1} + \dots$ for $m > 0$ and $y = y_0 + y_1 \zeta + \dots$, then direct calculation shows that we have
\begin{equation}
\bar{X}_0 = 0, \quad \bar{X}_1 = -\frac{1}{bm}, \quad \bar{X}_2 = \frac{b x_{m+1} - a x_m}{b^2 m^2 x_m} - \frac{y_0}{b^2 m^2}, \quad\dots
\end{equation}
and more generally that 
\begin{equation}
\begin{aligned}
\bar{X}_n &= \frac{P_n(x_m, \dots, x_{m+n}, y_0, \dots, y_{n-1})}{b^n m^n x_m^{n-1}} - \frac{y_{n-2}}{b^2 m^2}, \\
\bar{y}_n &= 0 \text{ for } n <m, \quad \bar{y}_n = x_m \text{ for } n \geq m,
\end{aligned}
\end{equation}
where $P_n$ is polynomial in its arguments. By again considering parametrisations and their Jacobians, it is straightforward to show that we cannot have blow-down singularities away from $x_m = 0$. Applying this argument inductively completes the proof that the only accessible blow-down singularities are as claimed.
\end{proof}

We now show how the singularity patterns pointed out in \autoref{infinitefamilies} correspond to confinement of blow-down singularities for equation \eqref{ddP1}. 
\begin{example} \label{eq1exampleB1}
The singularity $B_1$, which corresponds to the beginning of the singularity pattern
\begin{equation*}
(\operatorname{rg}, 0^1, \infty^1, \infty^1, 0^1, \operatorname{rg} ),
\end{equation*}
is confined after five iterations. We calculate as we did in \autoref{section2} but keep track of orders of jets and codimensions to find that composing the iteration on sections gives maps as follows:
\begin{align*}
~&B_1  \subset J^{r+5} &&\quad &&&\operatorname{codim}(B_1) = 1 &&
&& &&&&&(x^{(0)},y^{(0)}) = (0^1, \operatorname{rg})  \\
\varphi^{(1)} : ~& B_1 \rightarrow J^{r+5} &&\quad &&&\operatorname{codim}(\varphi^{(1)}(B_1)) = 3 &&
&& &&&&&(x^{(1)},y^{(1)}) = (\infty^1, 0^1)  \\
\varphi^{(2)} : ~& B_1 \rightarrow J^{r+5} &&\quad &&&\operatorname{codim}(\varphi^{(2)}(B_1)) = 5  &&
&& &&&&&(x^{(2)},y^{(2)})  = (\infty^1, \infty^1)\\
\varphi^{(3)} : ~& B_1 \rightarrow J^{r+3} &&\quad &&&\operatorname{codim}(\varphi^{(3)}(B_1)) = 3  &&
&& &&&&&(x^{(3)},y^{(3)})  = (0^1, \infty^1)\\
\varphi^{(4)} : ~& B_1 \rightarrow J^{r+1} &&\quad &&&\operatorname{codim}(\varphi^{(4)}(B_1)) = 1  &&
&& &&&&&(x^{(4)},y^{(4)}) = ( \operatorname{rg}, 0^1)\\
\varphi^{(5)} : ~& B_1 \rightarrow J^{r} &&\quad &&&\operatorname{codim}(\varphi^{(5)}(B_1)) = 0  &&
&& &&&&&(x^{(5)},y^{(5)})  = ( \operatorname{rg}, \operatorname{rg})
\end{align*}
For each iteration, we have indicated the order of jet space to which we have well-defined mappings from $B_1$, as well as codimensions of the images of $B_1$ and the corresponding parts of the singularity pattern. We note that the exceptional behaviour we observed in the singularity pattern, namely that when computing $x^{(3)}$, three terms vanished as it developed a zero rather than a pole, is reflected in the codimension falling from $5$ to $3$.
\end{example}

More generally, if we take the blow-down singularities $B_m$ as in \autoref{lemmaeq1} as subsets of $J^{2m+3+r}$ with the rest of the coefficients generic, from \autoref{theoremeq1} we see that iterating the system \eqref{ddP1} induces a map $\varphi^{(2m+3)} : B_m \rightarrow  J^{r}$, where the image of $B_m$ is a jet visible in the $(x,y)$ chart. To see that this image is of codimension zero, we must make some observations of how the initial data from the section $(x^{(0)},y^{(0)})$ enters into the subsequent iterates, and in particular how it is recovered in $(x^{(2m+3)}, y^{(2m+3)})$. This will require detailed but straightforward analysis of the mapping on jets in three cases, corresponding to different points in the singularity pattern. Firstly, when the first pole develops and how the coefficients from $(x^{(0)}, y^{(0)})$ enter into $X^{(1)}, X^{(2)}$, secondly, how the initial data is propagated through the sequence of simple poles $X^{(1)}, \dots, X^{(2m)}$, then how it reenters $x^{(2m+2)}, x^{(2m+3)}$ after the zero develops at $x^{(2m+1)}$. The key technique for our analysis here is essentially identifying and counting free variables, which we illustrate in detail in this example. \\

We first consider the map from $(x^{(0)}, y^{(0)})$ to $(X^{(1)},y^{(1)})$ corresponding to the development of the first simple pole in the sequence. Here we omit the superscripts for conciseness, working with the mapping in the charts $(x,y)$ and $(\bar{X},\bar{y})$. Beginning with initial data corresponding to $B_m$, namely sections in the $(x,y)$ chart with $x_0 = x_1 = \dots x_{m-1} = 0$, with the rest of the coefficients $x_i, y_j$ generic, by direct calculation we have 
\begin{equation*}
\begin{aligned}
\bar{X}_0 &=0, ~~\bar{X}_1= - \frac{1}{m b},  \quad &&\bar{X}_{n} =  -\frac{y_{n-2}}{m^2 b^2} + \frac{P_n(x_m,\dots, x_{m+n-1}, y_0 ,\dots, y_{n-3})}{x_m^{n-1}}, \quad \text{ for } n\geq 2\\
\bar{y}_0 &= \dots = \bar{y}_{m-1} = 0, \quad  &&\bar{y}_n = x_n, \quad \quad \text{ for } n\geq m,
\end{aligned}
\end{equation*}
where $P_n$ is polynomial in its arguments. From this, we see that the coefficients $\bar{X}_{i\geq 2}, \bar{y}_{j\geq m}$ are algebraically independent functions of the initial data, which follows from the way in which the free variable $y_{n-2}$ ($n\geq 2$) appears linearly in $\bar{X}_n$ but not at all in $\bar{X}_{n-1}$ and so on. In particular we have the image of $B_m$ under a single iteration being of codimension $m+2$, as noted in Example \autoref{exampleeq1sings}. Similarly, we see that the next iterate is obtained from $\bar{X}_i, \bar{y}_j$ above as
\begin{equation*}
\begin{aligned}
\bar{\bar{X}}_0 &=0, ~~\bar{\bar{X}}_1= \frac{1}{b},  \quad \bar{\bar{X}}_{j} =  P_j(\bar{X}_0, \dots, \bar{X}_j), \quad &&\text{ for } 2 \leq j \leq m+1, \\
&\bar{\bar{X}}_{n} =  -\frac{\bar{y}_{n-2}}{b^2} + Q_n(\bar{X}_0, \dots, \bar{X}_n, \bar{y}_m \dots \bar{y}_{n-3}), \quad &&\text{ for } n \geq m+2, \\
&\bar{\bar{Y}}_0 = 0, \quad \bar{\bar{Y}}_1 = - \frac{1}{m b }, \quad \bar{\bar{Y}}_n = \bar{X}_n , \quad &&\text{ for } n \geq 2.
\end{aligned}
\end{equation*}
Here $P_j$ is again polynomial, linear in $\bar{X}_j$, and $Q_n$ is polynomial in its arguments. From this, we see that the image of $B_m$ is of codimension $m+4$, with $\bar{\bar{X}}_i, \bar{\bar{Y}}_{j}$ having the following dependence on the initial data $x_i,y_j$: 
\begin{equation*}
\begin{aligned}
\bar{\bar{X}}_0 &=0, ~~\bar{\bar{X}}_1= \frac{1}{b}, &&\bar{\bar{X}}_{n} =  F_n( y_0, \dots, y_{n-2}, x_m, \dots, x_{m+n-1}) &&&\text{ for } n \geq m+2, \\
\bar{\bar{Y}}_0 &=0, ~~\bar{\bar{Y}}_1= -\frac{1}{m b},  &&\bar{\bar{Y}}_{n} =  G_n( y_0, \dots, y_{n-2}, x_m, \dots, x_{m+n-1}) &&&\text{ for } n \geq m+2,
\end{aligned}
\end{equation*}
where, importantly, $F_n$ is linear in $y_{n-2}$ with constant coefficient, and also linear in $x_{m+n-1}$ with coefficient being a constant multiple of $1/x_m$. \\

We now consider the iterates $X^{(3)}, \dots, X^{(2m)}$, which correspond to simple poles, and show that we have the same kind of dependence of coefficients on the initial data. Building on our calculation \eqref{poles} in the charts $(X,Y), (\bar{X},\bar{Y})$, we see that sections with $(X_0, Y_0) = (0,0)$ have images under the iteration mapping given by
\begin{equation} \label{simplepolerecs}
\begin{aligned}
\bar{X}_0 &=0, &&\bar{X}_1= \frac{Y_1}{1+b Y_1},  \quad \bar{X}_{n} =  \frac{Y_n}{(1+bY_1)^2} + \frac{P_n(X_1,\dots, X_{n}, Y_1, \dots Y_{n-1})}{X_1^{n-1}(1+b Y_1)^n},\\
\bar{Y}_0 &= 0, &&\bar{Y}_1 = X_1, \quad \quad \quad ~~\bar{Y}_n = X_n, \quad \quad \text{ for } n\geq 2,
 \end{aligned}
\end{equation}
where $P_n$ is polynomial in its arguments, and we note that these expansions are valid for determining all iterates $(X^{(2)}, Y^{(2)}), \dots, (X^{(2m)}, Y^{(2m)})$, as we have $X_1^{(k)} \neq 0, 1+b Y_1^{(k)} \neq 0$, for $k=0,\dots, 2m-1$, which we know from our explicit expressions of the residues of the simple poles in the singularity pattern, given in \autoref{theoremeq1} . Iterating through this sequence of simple poles, we have well-defined maps $J^{2m+3+r} \backslash \left\{ X_1 (1+b Y_1)= 0 \right\} \rightarrow J^{2m+3+r}$, and a simple calculation using the Jacobian as in the proof of \autoref{lemmaeq1} shows that the image of $B_m$ cannot change codimension in $J^{2m+3+r}$ under this sequence of maps, so we have the images of $B_m$ under $\varphi^{(2)}, \dots, \varphi^{(2m)}$ are all of codimension $m+4$. \\

Further, from \eqref{simplepolerecs} and our observations of $(\bar{\bar{X}}, \bar{\bar{Y}})$ we see that for $k = 2, \dots, 2m$, the coefficients $X^{(k)}_n, Y^{(k)}_n$ have the same kind of dependence on the initial data, and in particular the last iterate before the zero develops is of the form
\begin{equation*}
\begin{aligned}
X^{(2m)}_0 &=0, ~~X^{(2m)}_1= \frac{1}{mb}, &&X^{(2m)}_{n} =  F^{(2m)}_n( y_0, \dots, y_{n-2}, x_m, \dots, x_{m+n-1}) &&&\text{ for } n \geq m+2, \\
Y^{(2m)}_0 &=0, ~~Y^{(2m)}_1= -\frac{1}{b},  &&Y^{(2m)}_{n} =  G^{(2m)}_n( y_0, \dots, y_{n-2}, x_m, \dots, x_{m+n-1}) &&&\text{ for } n \geq m+2,
\end{aligned}
\end{equation*}
where again $F_n$ is linear in $y_{n-2}$ with constant coefficient, and also linear in $x_{m+n-1}$ with coefficient being a constant multiple of $1/x_m$.

We now consider the final step, when the map $(X^{(2m)}, Y^{(2m)}) \mapsto (x^{(2m+1)}, Y^{(2m+1)})$ shows a drop in codimension of the image of $B_m$, with the development of a zero of order $m$. Omiting the superscripts for conciseness and writing $(X^{(2m)}, Y^{(2m)}) = (X_0 + X_1 \zeta+\dots, Y_0+ Y_1 \zeta + \dots)$, we know that the the coefficients for the image of the $B_m$ under the iterations up to this point in the singularity pattern must satisfy at least 
\begin{equation}\label{endpattern}
Y_0 = 0, \quad Y_1 = - b^{-1}, \quad X_0 = 0, \quad X_1 = (mb)^{-1}.
\end{equation}
Similarly writing $(x^{(2m+1)}, Y^{(2m+1)}) = (\bar{x}_0 +\bar{x}_1 \zeta + \dots, \bar{Y}_0 + \bar{Y}_1 \zeta+ \dots)$, we see the mapping on coefficients from jets satisfying \eqref{endpattern} gives
\begin{equation} \label{laststep}
\begin{aligned}
&\bar{x}_0 = 0,\quad \bar{x}_1 = a + b^2 m X_2 - b^2 Y_2, \quad \bar{x}_2 =-b^2 \left(b m^2 X_2^2+b Y_2^2-2 m X_3+Y_3\right), \\
&\bar{x}_n = b^2(  n m X_{n+1} - Y_{n+1} ) + P_n (X_2,\dots, X_{n}, Y_2, \dots, Y_{n}), \quad \text{ for } n \geq 1,  \\
&\bar{Y}_0 = 0,  \quad \bar{Y}_1 = \frac{1}{mb}, \quad \bar{Y}_j = X_j, \quad \text{ for } j\geq 2,
\end{aligned}
\end{equation}
where we have again used $P_n$ to denote a polynomial in its arguments. We know from \autoref{theoremeq1} that if $(X^{(2m)}, Y^{(2m)})$ are obtained by iterating from $B_m$, then the coefficients $X_i, Y_j$ must satisfy the algebraic conditions for $\bar{x}_0, \dots, \bar{x}_{m-1}$ given by \eqref{laststep} to all vanish, and we know exactly what relations must exist between the coefficients $(X^{(2m)}_i, Y^{(2m)}_j)$, which have evolved through the singularity pattern from those defining $B_m$. Further, from the dependence of $X_i^{(2m)}, Y_j^{(2m)}$ on the initial data, and the way in which $X^{(2m)}_i, Y^{(2m)}_j$ enter into $x^{(2m+1)}_i , Y^{(2m+1)}_j$ according to \eqref{laststep}, we see that the image of $B_m$ under $\varphi^{(2m+1)}$ is of codimension $m+3$ in the jet space corresponding to $(x^{(2m+1)} , Y^{(2m+1)})$. Finally, another calculation on the exact same lines shows that after one more step, we have the image of $B_m$ under $\varphi^{(2m+1)}$ being of codimension zero.

\subsubsection{Equation (1.2)}

The analysis in this case proceeds in exactly the same way as the previous one, so we omit details for conciseness. In particular, the following may be proved using the same techniques and approach as for \autoref{lemmaeq1}:
\begin{lemma} \label{lemmaeq2}
The only accessible blow-down type singularities of equation \eqref{ddP1} are 
\begin{equation*}
B_m = \left\{ x_i = 0 ~~\forall i =0,\dots, m-1 \right\}.
\end{equation*}
\end{lemma}
We may also use the same techniques to examine the behaviour of blow-down singularities in terms of codimension, beginning with that associated with a simple zero:
\begin{example} \label{eq2exampleB1}
The singularity $B_1$ of equation \eqref{dHKdV}, which corresponds to the start of the singularity pattern
\begin{equation*}
(\operatorname{rg}, 0^1 ,\infty^2 , 0^1, \operatorname{rg} ),
\end{equation*}
is confined after four iterations, with the following behaviour under compositions of the iteration maps:
\begin{align*}
~&B_1  \subset J^{r+4} &&\quad &&&\operatorname{codim}(B_1) = 1 &&
&& &&&&&(x^{(0)},y^{(0)}) = (0^1, \operatorname{rg})  \\
\varphi^{(1)} : ~& B_1 \rightarrow J^{r+4} &&\quad &&&\operatorname{codim}(\varphi^{(1)}(B_1)) = 5 &&
&& &&&&&(x^{(1)},y^{(1)}) = (\infty^2, 0^1)  \\
\varphi^{(2)} : ~& B_1 \rightarrow J^{r+4} &&\quad &&&\operatorname{codim}(\varphi^{(2)}(B_1)) = 5  &&
&& &&&&&(x^{(2)},y^{(2)})  = (0^1, \infty^2)\\
\varphi^{(3)} : ~& B_1 \rightarrow J^{r+1} &&\quad &&&\operatorname{codim}(\varphi^{(3)}(B_1)) = 1  &&
&& &&&&&(x^{(3)},y^{(3)})  = (\operatorname{rg}, 0^1)\\
\varphi^{(4)} : ~& B_1 \rightarrow J^{r} &&\quad &&&\operatorname{codim}(\varphi^{(4)}(B_1)) = 0  &&
&& &&&&&(x^{(4)},y^{(4)}) = (\operatorname{rg}, \operatorname{rg})
\end{align*}
We note here again that the drop in codimension occurs when two terms vanish in the expansion for $x^{(3)}$ as it regains regularity as opposed to having a double pole.
\end{example}
Again, considering the blow-down singularities $B_m$ from \autoref{lemmaeq2} as subsets of $J^{2m+2+r}$, \autoref{theoremeq2} and tracing the dependence on initial data of the iterates through the sequence using exactly the same techniques as in the previous example, we see that we have $\varphi^{(2m+2)} : B_m \rightarrow  J^{r}$ under which the image of $B_m$ is of codimension zero, so all accessible blow-down singularities of equation \eqref{dHKdV} are confined.
\subsubsection{Equation (1.3)}

In this case we begin with an example, as the blow-down singularities for equation \eqref{eq20} occur not after $x$ develops a zero at $z_0$, but under the mapping applied to the jets in $\ubar{x},\ubar{y}$ coordinates satisfying the condition for a zero to develop.
\begin{example} \label{eq2exampleB1}
The condition on $(\underline{x}, \underline{y})$ for a simple zero to develop while iterating equation \eqref{eq20}, namely 
\begin{equation*}
B_1 = \left\{ \alpha \ubar{x}_1+\lambda(z_0 - 1) \ubar{x}_0 = 0 \right\},
\end{equation*}
with the rest of the coefficients generic, corresponds to the start of the singularity pattern which we denoted in \autoref{section2} by
\begin{equation*}
\left( \operatorname{rg}, \ubar{\zeta}_0^1,  0^{1} , \infty^{1}, \bar{\bar{\zeta}}_0^1 , \operatorname{rg} \right).
\end{equation*}
We observe a jump in codimension not from $(x^{(0)}, y^{(0)})$ to $(x^{(1)}, y^{(1)})$, but one step earlier, and we observe the following behaviour under compositions of the iteration maps:
\begin{align*}
~&B_1  \subset J^{r+5} &&\quad &&&\operatorname{codim}(B_1) = 1 &&
&& &&&&&(x^{(-1)},y^{(-1)}) = (\ubar{\zeta}_0^1, \operatorname{rg})  \\
\varphi^{(1)} : ~& B_1 \rightarrow J^{r+4} &&\quad &&&\operatorname{codim}(\varphi^{(1)}(B_1)) = 2 &&
&& &&&&&(x^{(0)},y^{(0)}) = (0^1, \ubar{\zeta}_0^1)  \\
\varphi^{(2)} : ~& B_1 \rightarrow J^{r+4} &&\quad &&&\operatorname{codim}(\varphi^{(2)}(B_1)) = 2  &&
&& &&&&&(x^{(1)},y^{(1)})  = (\infty^1,0^1)\\
\varphi^{(3)} : ~& B_1 \rightarrow J^{r+3} &&\quad &&&\operatorname{codim}(\varphi^{(3)}(B_1)) = 2  &&
&& &&&&&(x^{(2)},y^{(2)})  = ( \bar{\bar{\zeta}}_0^1, \infty^1)\\
\varphi^{(4)} : ~& B_1 \rightarrow J^{r+1} &&\quad &&&\operatorname{codim}(\varphi^{(4)}(B_1)) = 1  &&
&& &&&&&(x^{(3)},y^{(3)}) = (\operatorname{rg},  \bar{\bar{\zeta}}_0^1)\\
\varphi^{(5)} : ~& B_1 \rightarrow J^{r} &&\quad &&&\operatorname{codim}(\varphi^{(5)}(B_1)) = 0  &&
&& &&&&&(x^{(4)},y^{(4)}) = (\operatorname{rg},  \operatorname{rg})
\end{align*}
We note here again that a drop in codimension occurs when $x^{(3)}$ regains regularity as opposed to a simple zero.
\end{example}
Again by the same approach, the following may proved by local calculations in charts:
\begin{lemma} \label{lemmaeq3}
The only accessible blow-down type singularities of equation \eqref{eq20} are 
\begin{equation*}
B_m = \left\{ \frac{d^i}{dz^i} \left. \left(\lambda z \ubar{x}(z) + \alpha \ubar{x}'(z) \right) \right|_{z=z_0}= 0, \quad \forall i = 0, \dots, m-1 \right\}
\end{equation*}
\end{lemma}
In the same way as the other two examples, we see from \autoref{theoremeq3} that for regarding $B_m$ as a subset of $J^{(2m+2)}$, iterating the system gives a map $\varphi^{(2m+3)} : J^{(2m+3+r)} \rightarrow  J^{(r)}$, under which the image of $B_m$ is of codimension zero.

\section{Conclusions}

We now summarise our work and discuss questions that follow it naturally, again organised into two parts: firstly singularity analysis on the level of equations and secondly its geometric interpretation. On this first level, we have significantly extended previous studies of delay Painlev\'e equations and discovered new confinement type behaviour, which is interesting in its own right. In the process we have developed techniques for the analysis of singularity patterns of arbitrary length and proving confinement, which we hope will be useful in tackling one of the main difficulties in the singularity analysis of delay-differential equations. It would be interesting to adapt our methods to other integrable delay-differential equations, for example extensions of the examples considered in this paper such as the families generalising equation \eqref{dHKdV} isolated by Halburd and Korhonen by imposing Nevanlinna-theoretic integrability criteria \cite{RODRISTO}. Though preliminary calculations show that these equations admit some of the same confined singularity patterns as equation \eqref{dHKdV} (namely those associated with single, double and triple zeroes) it is a natural next step to determine whether these admit the same infinite families and whether this behaviour fits into our geometric framework.\\

Another question that arises from our work on the level of equations relates to the use of singularity analysis techniques to isolate integrability candidates. The fact that each of these three examples may be obtained by requiring confinement of only the simplest singularity in the family associated with zeroes of different orders prompts the question of whether and how this could ensure confinement of all singularities in the family. Further, there may be applications of our results to the search for elliptic function solutions of degenerate cases of delay Painlev\'e equations. For example, the $a=0$ and $p=0$ cases of equations \eqref{ddP1} and \eqref{dHKdV} respectively are known \cite{BJORNTHESIS} to admit elliptic function solutions. Degree 2 elliptic function solutions were identified with the help of singularity analysis, and in particular that these degenerate cases admit the singularity patterns associated with simple zeroes outlined in \autoref{section2}. These patterns are compatible with elliptic function solutions in the sense that the numbers of poles and zeroes in a pattern are equal (counted with multiplicity), and also that the residues of poles in the sequence sum to zero. We note that our proofs of the infinite families of singularity patterns are also valid for the degenerate cases, and we observe the same kind of compatibility with elliptic functions in all of them, so it would be interesting to determine whether they may be used to isolate higher degree elliptic function solutions.\\

The other aim of this work was to initiate the geometric study of delay Painlev\'e equations. We have put forward a geometric description of singularity confinement in these three examples, and we hope to have worked in convincing parallel with the discrete case, and in particular captured in our description the exceptional nature of these equations in terms of the recovery of initial data when a singularity is confined. By no means, however, is this geometric framework complete or definitive, and we hope that our ideas are refined and built upon through singularity analysis in more examples.

\subsection*{Acknowledgements}
The author would like to express his sincere thanks to R. Halburd for valuable discussions and advice.This research was supported by a University College London Graduate Research Scholarship and Overseas Research Scholarship.

\appendix

\section{Proofs of infinite families of singularity patterns}
We now give proofs of the results of \autoref{infinitefamilies} relating to infinite families of singularity patterns.

\subsection{Proof of Theorem 2.2} \label{proofeq1}
For equation \eqref{ddP1}, our strategy is to consider a singularity pattern beginning with $(\operatorname{rg}, 0^m)$, then derive and analyse recurrences for the coefficients in the expansions of the next $(2m+1)$ iterates, to deduce that the singularity pattern is as claimed. \\

Because the equation \eqref{ddP1} is autonomous we can take without loss of generality the zero of order $m$ to be at the origin, and start with the formal expansions
\begin{subequations} \label{initialexpansions}
\begin{align}
\ubar{u} &= \sum_{j=0}^{\infty} \ubar{u}_j z^j, \\
u &= \sum_{j=m}^{\infty} u_{j} z^{j}, \quad \quad u_m \neq 0.
\end{align}
\end{subequations}
Inserting these into the equation, we immediately see that $\bar{u}$ has a simple pole:
\begin{equation}
\bar{u} = - \frac{m \beta}{z} + O(1).
\end{equation}
The iterates of interest to us are $u= u^{(0)}, \bar{u}=u^{(1)}, u^{(2)}, \dots, u^{(2m)}, u^{(2m+1)}$. By inspection of the terms on the right-hand side of the forward iteration \eqref{ddP1fwd}, these will be either regular or poles of order at most one, so we introduce the notation
\begin{equation} \label{expansion}
u^{(i)} = \sum_{n=-1}^{\infty} u^{(i)}_{n} z^n,
\end{equation}
for $i=0, \dots, 2m+1$, where any number of the $u^{(i)}_{n}$ may be zero. \\

By deriving recurrences for the coefficients $u^{(i)}_n$, we will show firstly that $u_{-1}^{(i)} \neq 0$ for $i=1, \dots, 2m$, then that $u_{-1}^{(2m+1)} = u_{0}^{(2m+1)} = \dots = u_{m-1}^{(2m+1)} = 0$, from which we will deduce that $u^{(2m+1)} = \mathcal{O}(z^m)$, and in particular has a zero of order $m$ if the rest of the initial data is generic. \\

It will be helpful to introduce some notation to deal with the logarithmic derivative $u'/u$ in the forward iteration map.
\begin{lemma} \label{logderiv}
Let $r$ be a nonzero integer. If $u = \sum_{j=r}^{\infty} u_j z^j$ with $u_r$ nonzero, then 
\begin{equation*}
\frac{u'}{u} = \sum_{n=-1}^{\infty} U_{n} z^{n} ,
\end{equation*}
where the coefficients $U_n$ are given by $U_{-1} = r, ~U_0 = u_{r+1}/u_r,$ and so on according to the recurrence
\begin{equation*}
U_n =  \frac{1}{u_r} \left(  (n+1) u_{r+n+1} - \sum_{j=1}^{n} u_{r+j} U_{n-j} \right).
\end{equation*}
\end{lemma}
We first deduce from the recurrence that following the zero of order $m$, the next $2m$ iterates have simple poles:
\begin{proposition} \label{residues}
The iterates $u^{(i)}$ have simple poles at $z=0$ for all $i=1, \dots, 2m$, and we have 
\begin{subequations} \label{uminus}
\begin{align}
u_{-1}^{(2k)} &= k \beta, &&\text{for } k=1, \dots, m,\\
u_{-1}^{(2k+1)} &= (k-m)\beta &&\text{for } k=0, \dots, m.
\end{align}
\end{subequations}
\end{proposition}
\begin{proof}
We already have that $u^{(0)}_{-1} = 0$ and $u^{(1)}_{-1}= -m\beta$. We then insert the expansions \eqref{expansion} for the iterates $u^{(i)}$ into the relevant upshifts of the equation, making use of Lemma \ref{logderiv} with $r=-1$, which gives
\begin{equation}
u_{-1}^{(i+1)} = u_{-1}^{(i-1)} + \beta,
\end{equation}
for all $i$ such that $u^{(i)}$ has a simple pole. Iterating this from $i=1$ from the initial values for $u^{(0)}_{-1}, u^{(1)}_{-1}$, we see that $u^{(i)}$ have simple poles for all $i=1, \dots, 2m$, and we obtain the formulae \eqref{uminus}.
\end{proof}
It will now be helpful to introduce the following notation for the iterates:
\begin{align}
u^{(2k)} &= f^{(k)} = \sum_{n=-1}^{\infty} f_n^{(k)}z^n, && f_n^{(k)} = u_n^{(2k)}, \\
u^{(2k+1)} &= g^{(k)} = \sum_{n=-1}^{\infty} g_n^{(k)}z^n, && g_n^{(k)} = u_n^{(2k+1)},
\end{align}
for $k=0, \dots, m$.
As we now know that $u^{(1)}, \dots, u^{(2m)}$ have simple poles at $z=0$, we use Lemma \ref{logderiv} to write the logarithmic derivatives of $f^{(k)}, g^{(k-1)}$, for $k=1,\dots m$ as 
\begin{equation} \label{Us}
\frac{f^{(k)'}}{f^{(k)}} = \sum_{n=-1}^{\infty} F^{(k)}_n z^n, \quad \quad \frac{g^{(k)'}}{g^{(k)}} = \sum_{n=-1}^{\infty}G^{(k)}_n z^n.
\end{equation}
Further, we have from \eqref{uminus} that for $k=0, \dots, m$ that $F_{-1}^{(k)} = k \beta, ~G_{-1}^{(k)} = (m-k) \beta$, so we have the following recursive formulae for $F_n^{(k)}, G_n^{(k)}$:
\begin{subequations} \label{FGrecs}
\begin{align}
F_n^{(k)} &=  \frac{1}{k\beta} \left(  (n+1) f^{(k)}_{n} - \sum_{j=1}^{n} f^{(k)}_{j-1} F^{(k)}_{n-j}\right), \\
G_n^{(k)} &=  \frac{1}{(k-m)\beta} \left( (n+1) g^{(k)}_{n} - \sum_{j=1}^{n} g^{(k)}_{j-1} G^{(k)}_{n-j} \right), 
\end{align}
\end{subequations}
valid for all $k$ such that $f^{(k)}, g^{(k)}$ have simple poles. Using this notation, the forward iteration then leads to the recurrences,
\begin{subequations} \label{fg0rec}
\begin{align}
f_{0}^{(k)} &= f_{0}^{(k-1)} + \alpha - \beta G_0^{(k-1)}, \\
g_{0}^{(k)} &= g_{0}^{(k-1)} + \alpha - \beta F_0^{(k)},
\end{align}
\end{subequations}
and 
\begin{subequations} \label{fgnrec}
\begin{align}
f_{n}^{(k)} &= f_{n}^{(k-1)}  - \beta G_n^{(k-1)}, \\
g_{n}^{(k)} &= g_{n}^{(k-1)}  - \beta F_n^{(k)},
\end{align}
\end{subequations}
for $n \geq 1$, and $k=1,\dots, m$. Using \eqref{FGrecs} with $n=0$, we see that the recurrences \eqref{fg0rec} are a linear system of difference equations for $f_0^{(k)}, g_0^{(k)}$:
\begin{subequations} \label{fg0rec2}
\begin{align}
f_{0}^{(k)} &= f_{0}^{(k-1)} + \alpha - \frac{1}{(k-1-m)} g_0^{(k-1)}, \\
g_{0}^{(k)} &= g_{0}^{(k-1)} + \alpha - \frac{1}{k} f_0^{(k)},
\end{align}
\end{subequations}
subject to the initial conditions $f_0^{(0)} = 0$ and $g_0^{(0)} = u_0^{(1)} = \alpha+ \ubar{u}_0 - \beta u_1/u_0$ determined  by the initial data $\ubar{u}, u$. The unique solution of \eqref{fg0rec2} subject to these initial conditions is given by
\begin{subequations} \label{u0sol}
\begin{align}
f_0^{(k)} &= k (\alpha+C), \\
g_0^{(k)} &= (m-k)C,
\end{align}
\end{subequations}
where $C = u_0^{(1)}/m$.
Similarly, after using the formula \eqref{FGrecs}, the recurrences \eqref{fgnrec} become
\begin{subequations} \label{unrec2}
\begin{align}
f_{n}^{(k)} &= f_{n}^{(k-1)}  - \frac{1}{(k-1-m)} \left( (n+1) g^{(k-1)}_{n} - \sum_{j=1}^{n} g^{(k-1)}_{j-1} G^{(k-1)}_{n-j} \right), \\
g_{n}^{(k)} &= g_{n}^{(k-1)} - \frac{1}{k}  \left( (n+1) f^{(k)}_{n} - \sum_{j=1}^{n} f^{(k)}_{j-1} F^{(k)}_{n-j} \right), 
\end{align}
\end{subequations}
subject to the initial conditions $f_n^{(0)}=0$ for $n=0,\dots, m-1$, and $g_n^{(0)}$ fixed by the initial data $\ubar{u}, u$. Given the solution \eqref{u0sol}, the $n=1$ case is then a linear system of recurrences in $k$ for $f_1^{(k)}, g_1^{(k)}$, which may be solved by elementary methods. With both $n=0,1$ solutions in hand the system for $f_2^{(k)}, g_2^{(k)}$ can be solved, and so on. Observations of these solutions lead us to the following proposition:

\begin{proposition}
The unique solution to \eqref{unrec2} subject to the initial conditions is given by $(f_n^{(k)},g_n^{(k)})$, $n=0, \dots, m-1, k= 0,\dots, m$ of the form
\begin{subequations} \label{unsol}
\begin{align}
f_n^{(k)} &= k P_n^{(k)} \\
g_n^{(k)} &= (k-m)Q_n^{(k)},
\end{align}
\end{subequations}
where $P_n^{(k)}, Q_n^{(k)}$ are polynomial in $k$ of degree at most $n$.
\end{proposition}
\begin{proof}
We have from formulae \eqref{u0sol} that the statement is true for $n=0$, so we proceed by induction. Suppose that $f_0^{(k)}, \dots, f_{n-1}^{(k)}$ and $g_0^{(k)}, \dots, g_{n-1}^{(k)}$ are of the form \eqref{unsol}. The recursive formulae \eqref{FGrecs} then imply that $F_0^{(k}, \dots, F_{n-1}^{(k)}$ and $G_0^{(k)}, \dots, G_{n-1}^{(k)}$ are polynomial in $k$, of degree at most $n-1$. We then see that the following terms from \eqref{unrec2} are polynomial in $k$ of degree at most $n-1$: 
\begin{subequations}
\begin{align}
\frac{1}{k} \sum_{j=1}^{n} f^{(k)}_{j-1} F^{(k)}_{n-j} &= \sum_{j=1}^{n} P^{(k)}_{j-1} F^{(k)}_{n-j} = \sum_{j=0}^{n-1} \lambda_j k^j , \\
 \frac{1}{(k-m)} \sum_{j=1}^{n} g^{(k)}_{j-1} G^{(k)}_{n-j} &= \sum_{j=1}^{n} Q^{(k)}_{j-1} G^{(k)}_{n-j} = \sum_{j=0}^{n-1} \mu_j k^j ,
\end{align}
\end{subequations}
so we have
\begin{subequations} \label{uneqs}
\begin{align}
f_{n}^{(k+1)} &= f_{n}^{(k)}  - \frac{n+1}{k-m} g^{(k)}_{n} + \sum_{j=0}^{n-1} \mu_j k^j, \\
g_{n}^{(k+1)} &= g_{n}^{(k)} - \frac{n+1}{k+1} f^{(k)}_{n} \sum_{j=0}^{n-1} \lambda_j k^j , 
\end{align}
\end{subequations}
We write our ansatz \eqref{unsol} for the solution to this equation as 
\begin{subequations} \label{unansatz}
\begin{align}
f_n^{(k)} &= k \sum_{j=0}^{n} a_j k^j\\
g_n^{(k)} &= (k-m)\sum_{j=0}^{n} b_j k^j.
\end{align}
\end{subequations}
We note that one initial condition $u_n^{(0)} = 0$ is satisfied automatically, but imposing the other requires us to set
\begin{equation}
b_0 = - u_n^{(1)}/m.
\end{equation}
We now insert the ansatz \eqref{unansatz} into the equation \eqref{uneqs} and equate coefficients of powers of $k$ to obtain a linear system in $2n+1$ variables $a_0, \dots, a_n, b_1, \dots, b_n$:
\begin{subequations} \label{linearsystem}
\begin{align}
0 &= a_n + b_n, \\
\lambda_i &= (n+1) a_i + (i+1)b_i + \sum_{j=i+1}^n (-1)^{j-i}\left( \left(\begin{array}{c}j +1\\i\end{array}\right) +m \left(\begin{array}{c}j \\i\end{array}\right) \right) b_j , \\
\mu_i &= \sum_{j=i}^n \left(\begin{array}{c}j \\i\end{array}\right) a_j + (n+1) b_i, \\
\lambda_0 &= (n+1) a_0 +(m+1) \sum_{j=1}^n (-1)^j b_j,\\
\mu_0 &= \sum_{j=0}^n a_j +(n+1) b_0, 
\end{align}
\end{subequations}
for $i=1, \dots, n-1$. We write this as $\mathcal{M}_n \mathbf{v} = \mathbf{c}$, where $\mathbf{v} = (a_0, \dots, a_n, b_1, \dots, b_n)^T, \mathbf{c} = (0, \lambda_0, \dots, \lambda_n, \mu_0, \dots, \mu_n)^T$ and $\mathcal{M}_n$ is the square matrix of size $2n+1$ giving the right-hand side of the system \eqref{linearsystem}. A simple sequence of row and column operations yields an upper-triangular matrix and we obtain
\begin{equation}
\det \mathcal{M}_n = (n!)^2 (m-n)_n,
\end{equation}
where $(a)_n = \prod_{i=0}^{n-1} (a+i)$ is the usual Pochammer symbol, so the matrix is nonsingular for $n\leq m-1$, showing that the unique solution of the recurrence \eqref{uneqs} is of the form \eqref{unansatz} and the inductive step is complete.  
\end{proof}
Together with Proposition \ref{residues}, this allows us deduce that $u_n^{(2m+1)}=0$ for $m \geq n$, and thus that $u^{(2m+1)} = O(z^m)$. 
\subsection{Proof of Theorem 2.1} \label{proofeq2}
While we may proceed along the same lines as in \autoref{proofeq1}, a shortcut is provided by a known Miura-type transformation between equation \eqref{ddP1} and equation \eqref{dHKdV}. This may be easily detected given the well-known transformation between the differential-difference systems that give these equations as similarity reductions, and is proved by direct calculation:
\begin{lemma}
If $v$ solves \eqref{dHKdV} with parameters $p,q$, then $u = \ubar{v} v$ solves \eqref{ddP1} with parameters $a = 2p, b = -q$.
\end{lemma}
So, we consider a singularity pattern for equation \eqref{dHKdV} beginning with $(\ubar{v}, v) = (\operatorname{rg},0^m)$, and we also assume that the zero has developed while iterating through regular and nonzero iterates, so $\ubar{\ubar{v}}$ is also regular. Then under the transformation to a solution of \eqref{ddP1}, we have 
\begin{equation*}
(\ubar{u}, u ) = (\ubar{\ubar{v}} \ubar{v}, \ubar{v} v ) = (\operatorname{rg}, 0^m),
\end{equation*} 
so the transformation gives us a singularity pattern for \eqref{ddP1}, which by \autoref{theoremeq1} must be
\begin{equation*}
\left( \operatorname{rg}, 0^m , \infty^1, \infty^1, \infty^{1} , \dots, \infty^1 , \infty^1, 0^m, \operatorname{rg} \right),
\end{equation*}
with $u^{(k)} \sim \zeta^{-1}$ for $k=1, \dots, 2m$, then $u^{(2m+1)} \sim \zeta^m$, and $u^{(2m+2)}$ regular. So this implies that the iterates $v^{(k)}$ in the singularity pattern must satisfy:
\begin{subequations}
\begin{align}
v^{(k-1)} v^{(k)} &= u^{(k)} \sim \zeta^{-1} \quad \text{ for } k=1, \dots, 2m, \label{miuraorderseq2}\\
v^{(2m)} v^{(2m+1)} &= u^{(2m+1)} \sim \zeta^m ,\label{miuraorderseq2m} \\
v^{(2m+1)} v^{(2m+2)} &= u^{(2m+2)} = \mathcal{O}(\zeta^0).\label{miuraorderseq20}
\end{align}
\end{subequations}
Beginning with our assumption that $v^{(0)} \sim \zeta^m$, we see from the $k=1$ case of equation \eqref{miuraorderseq2} that $v^{(1)} \sim \zeta^{-(m+1)}$, and then using the $k=2, \dots, 2m$ cases successively that 
\begin{equation}
v^{(2k)} \sim \zeta^m, \text{ for } k = 0, \dots, m, \quad \quad v^{(2k+1)} \sim \zeta^{-(m+1)} \text{ for } k = 0, \dots, m-1.
\end{equation}
Then using equations \eqref{miuraorderseq2m} and \eqref{miuraorderseq20} we have that $v^{(2m+1)} \sim \zeta^0$ and $v^{(2m+2)} = \mathcal{O}(\zeta^0)$ and the proof is complete.

\subsection{Proof of Theorem 2.3}
Again, while the strategy and techniques from the proof of \autoref{theoremeq1} are available for this case, a shortcut is provided by the following transformation between equation \eqref{ddP1} and equation \eqref{eq20}, which was pointed out in \cite{RGMOREIRA}:
\begin{lemma}
If $w$ solves \eqref{eq20} with parameters $\lambda, \alpha$, then $u = \bar{w}/\ubar{w}$ solves \eqref{ddP1} with parameters $a= 2 \lambda, b = - \alpha$.
\end{lemma}
Similarly to in the previous section, we consider a singularity pattern for equation \eqref{eq20} beginning with $(\ubar{w}, w)$, where $\frac{d^i}{dz^i} \left(\lambda z \ubar{w}(z) + \alpha \ubar{w}'(z) \right) = 0$ at $z=z_0$ for $i=0,\dots, m-1$ and $w \sim \zeta^m$. and we also assume that the zero has developed while iterating through regular and nonzero iterates, so $\ubar{\ubar{w}}, \ubar{\ubar{\ubar{w}}}$ are also regular and nonzero. Then under the transformation to a solution of \eqref{ddP1}, we have 
\begin{equation*}
(\ubar{\ubar{u}},\ubar{u}) = (\frac{\ubar{w}}{\ubar{\ubar{\ubar{w}}}}, \frac{w}{\ubar{\ubar{w}}}) = (\operatorname{rg}, 0^m),
\end{equation*} 
so the transformation gives us a singularity pattern for \eqref{ddP1}, which by \autoref{theoremeq1} must be
\begin{equation*}
\left( \operatorname{rg}, 0^m , \infty^1, \infty^1, \infty^{1} , \dots, \infty^1 , \infty^1, 0^m, \operatorname{rg} \right),
\end{equation*}
with $u^{(k)} \sim \zeta^{-1}$ for $k=0, \dots, 2m-1$, then $u^{(2m)} \sim \zeta^m$, and $u^{(2m+2)}$ regular. So this implies that the iterates $w^{(k)}$ in the singularity pattern must satisfy:
\begin{subequations}
\begin{align}
\frac{w^{(k+1)}}{w^{(k-1)}}&= u^{(k)} \sim \zeta^{-1} \quad \text{ for } k=0, \dots, 2m-1, \label{miuraorderseq3}\\
\frac{w^{(2m+1)}}{w^{(2m-1)}} &= u^{(2m)} \sim \zeta^m ,\label{miuraorderseq3m} \\
\frac{w^{(2m+2)}}{w^{(2m)}} &= u^{(2m+1)} = \mathcal{O}(\zeta^0).\label{miuraorderseq30}
\end{align}
\end{subequations}
Beginning with our assumptions that $w^{(0)} \sim \zeta^m$ and $w^{(-1)}$ is regular, we see recursively from equation \eqref{miuraorderseq3} that
\begin{equation}
w^{(2k)} \sim \zeta^{m-k} \text{ for } k=0, \dots m-1, \quad \quad w^{(2k+1)} \sim \zeta^{-k} \text{ for } k=0, \dots,  m.
\end{equation}
Then using equations \eqref{miuraorderseq3m} and \eqref{miuraorderseq30} we see that $w^{(2m+1)} \sim \zeta^0$ and $w^{(2m+2)} = \mathcal{O}(\zeta^0)$ and the proof is complete.

\bibliographystyle{amsalpha}
\bibliography{delaybibliography}

\end{document}